\documentclass[12pt]{amsart}

\usepackage[T1]{fontenc}
\usepackage{kpfonts}

\usepackage{setspace}
\usepackage[margin=1.25in]{geometry}

\usepackage[dvipsnames]{xcolor}
\usepackage{pdfpages}
\usepackage{float}
\usepackage{multirow}
\usepackage{caption}
\usepackage{subcaption}
\usepackage{epigraph}

\usepackage{mathtools}
\usepackage{amsthm}
\usepackage{aliascnt} % needed so cleveref distinguishes theorem-like environments
\usepackage{accents}
\usepackage{dutchcal}
\usepackage{bbm}

\usepackage{enumitem}
\usepackage{sgame}
\usepackage{tikz-cd}
\usetikzlibrary{calc,shapes,arrows}

\usepackage{tcolorbox}
\usepackage{listings}
\usepackage[normalem]{ulem} % for \sout, used below in \hsout

\usepackage[round]{natbib}

\DeclareMathOperator{\supp}{supp}

\newcommand{\R}{\mathbb{R}}

\newcommand{\trans}{\text{trans}}

\newtheorem{theorem}{Theorem}[section]

\newaliascnt{proposition}{theorem}
\newtheorem{proposition}[proposition]{Proposition}
\aliascntresetthe{proposition}

\newaliascnt{lemma}{theorem}
\newtheorem{lemma}[lemma]{Lemma}
\aliascntresetthe{lemma}

\newaliascnt{corollary}{theorem}
\newtheorem{corollary}[corollary]{Corollary}
\aliascntresetthe{corollary}

\newaliascnt{claim}{theorem}

\aliascntresetthe{claim}

\theoremstyle{definition}

\newtheorem{definition}{Definition}

\newaliascnt{example}{theorem}
\newtheorem{example}[example]{Example}
\aliascntresetthe{example}

\newaliascnt{assumption}{theorem}
\newtheorem{assumption}[assumption]{Assumption}
\aliascntresetthe{assumption}

\newaliascnt{condition}{theorem}

\aliascntresetthe{condition}

\newaliascnt{question}{theorem}

\aliascntresetthe{question}

\newaliascnt{remark}{theorem}
\newtheorem{remark}[remark]{Remark}
\aliascntresetthe{remark}

\newaliascnt{remarks}{theorem}

\aliascntresetthe{remarks}

\newaliascnt{aside}{theorem}

\aliascntresetthe{aside}

\newaliascnt{note}{theorem}

\aliascntresetthe{note}

\usepackage{hyperref}

\hypersetup{
    colorlinks=true,
    linkcolor=CadetBlue,
    filecolor=CadetBlue,
    urlcolor=CadetBlue,
    citecolor=Mulberry,
}

\usepackage[nameinlink]{cleveref}

\crefname{theorem}{theorem}{theorems}
\Crefname{theorem}{Theorem}{Theorems}

\crefname{proposition}{proposition}{propositions}
\Crefname{proposition}{Proposition}{Propositions}

\crefname{lemma}{lemma}{lemmas}
\Crefname{lemma}{Lemma}{Lemmas}

\crefname{corollary}{corollary}{corollaries}
\Crefname{corollary}{Corollary}{Corollaries}

\crefname{claim}{claim}{claims}
\Crefname{claim}{Claim}{Claims}

\crefname{definition}{definition}{definitions}
\Crefname{definition}{Definition}{Definitions}

\crefname{example}{example}{examples}
\Crefname{example}{Example}{Examples}

\crefname{assumption}{assumption}{assumptions}
\Crefname{assumption}{Assumption}{Assumptions}

\crefname{condition}{condition}{conditions}
\Crefname{condition}{Condition}{Conditions}

\crefname{question}{question}{questions}
\Crefname{question}{Question}{Questions}

\crefname{remark}{remark}{remarks}
\Crefname{remark}{Remark}{Remarks}

\crefname{remarks}{remarks}{remarks}
\Crefname{remarks}{Remarks}{Remarks}

\crefname{aside}{aside}{asides}
\Crefname{aside}{Aside}{Asides}

\crefname{note}{note}{notes}
\Crefname{note}{Note}{Notes}

\definecolor{backcolour}{rgb}{0.63, 0.79, 0.95}

\lstdefinestyle{mystyle}{
  backgroundcolor=\color{backcolour},
  basicstyle=\ttfamily\footnotesize,
  breakatwhitespace=false,
  breaklines=true,
  captionpos=b,
  keepspaces=true,
  numbers=left,
  numbersep=5pt,
  showspaces=false,
  showstringspaces=false,
  showtabs=false,
  tabsize=2
}

\begin{document}
\title{Pass-through with Price Dispersion}
\author{Brian C.~Albrecht \and Mark Whitmeyer}
\thanks{BA: International Center for Law \& Economics. Email: \href{mailto:mail@briancalbrecht.com}{mail@briancalbrecht.com}. MW: Arizona State University. Email: \href{mailto:mark.whitmeyer@gmail.com}{mark.whitmeyer@gmail.com}. We thank Rani Spiegler and Joseph Whitmeyer for their feedback. \href{https://www.refine.ink/}{refine.ink} helped us improve this paper.}
\begin{abstract}
How do cost shocks pass through to prices in markets with price dispersion? We decompose the problem into two layers. In the competition layer, consumers' consideration sets determine equilibrium distributions of normalized margins. In the curvature layer, demand elasticity maps these margins into prices and pass-through rates. We prove the pricing game is strategically equivalent to a game over normalized margins, so equilibrium margin distributions are invariant to demand and costs. This separation yields closed-form pass-through formulas at each quantile of the price distribution, robust bounds across demand specifications, and sharp comparative statics linking market structure to incidence.
\end{abstract}
\maketitle

\section{Introduction}

Empirical studies consistently document substantial price variation for homogeneous products.\footnote{For evidence on dispersion in relatively homogeneous retail markets, see \citet{Sorensen2000Drugs}. For evidence from online price-comparison settings, see \citet{baye_morgan_scholten2004_price_dispersion_small_large}.} The theoretical toolkit for analyzing pass-through presumes a unique price to differentiate. This paper develops a tractable framework for pass-through analysis when equilibrium features a price distribution rather than a price point.

The core contribution is a decomposition. We show that pass-through in dispersed-price markets separates into two tiers. The first is a \textit{competition layer}, where market structure determines an equilibrium distribution of normalized effective margins. The second is a \textit{curvature layer}, where demand elasticity determines how these margins translate into prices and pass-through rates. Our key result, \Cref{thm:isomorphism}, establishes that \textbf{the equilibrium margin distribution depends only on market structure, not on demand curvature or cost levels}. Demand and costs matter only for the mechanical translation from margins to prices.

This decomposition has three immediate payoffs for applied work. First, it yields closed-form pass-through formulas at each quantile of the price distribution, computable without resolving the full equilibrium for each cost level. Second, the invariance of margin distributions to costs means pass-through follows from differentiating a function, rather than from comparative statics on a complex game. Third, the separation identifies exactly when demand specification matters. The competition layer is demand-free, while the curvature layer requires knowing the demand function.

The paper contributes to two literatures that have developed largely in parallel. The literature on pass-through, following  \citet{WeylFabinger2013PassThroughTool}, derives simple formulas linking pass-through to demand curvature but requires a single equilibrium price, either under monopoly or symmetric oligopoly. We show how their curvature insights extend to equilibrium price distributions. The same demand fundamentals matter, but they interact with a distribution of markups rather than a single markup.\footnote{For search or limited-information models where equilibrium markup distributions are central (even with homogeneous goods), see, e.g., \citet{menzio2024_markups_wp,kaplan_menzio_rudanko_trachter2019_relative_price_dispersion,menzio_trachter2015_equilibrium_price_dispersion_sequential_search,burdett_menzio2018_qss_pricing_rule}.} The modern industrial organization literature on price dispersion, building on \citet{varian1980model} and \citet{BurdettJudd1983EquilibriumPriceDispersion}, characterizes when and why prices are dispersed, but has not addressed how cost shocks transmit through these distributions. We provide a general framework of pass-through characterization for dispersed-price equilibria.

We adopt the consideration-set framework of \citet{ArmstrongVickers2022Patterns} as our model of price dispersion. Consumers observe prices only from firms in their consideration set and purchase from the cheapest option. This generates mixed-strategy equilibria where firms randomize over prices. The consideration structure (who considers whom) is reduced-form and captures many phenomena, including search costs, informational frictions, platform algorithms, and behavioral inattention.\footnote{For models where platforms/intermediaries shape consumer consideration and can steer or limit access, see \citet{nocke_rey2024_consumer_search_steering_choice_overload,dinerstein_einav_levin_sundaresan2018_consumer_price_search_platform_design,hagiu_jullien2011_why_intermediaries_divert_search,teh_wright2022_intermediation_steering,teh_wang_watanabe2024_strategic_limitation_market_accessibility}. For search-engine monetization or position effects, see \citet{eliaz_spiegler2011_simple_model_search_engine_pricing,athey_ellison2011_position_auctions_consumer_search}.}
Our results apply regardless of the underlying source.

The isomorphism we identify reformulates the pricing game in terms of normalized effective margins rather than prices. The effective margin captures profit per customer served, normalized to lie in the unit interval. Under a natural monotonicity condition, there is a bijection between prices and margins. The key observation is that equilibrium behavior depends on the trade-off between margin and market share, and this trade-off is determined entirely by the consideration structure. Demand curvature and costs \textit{affect only the translation from margins back to prices}. This is why the equilibrium margin distribution is invariant. When costs change, the entire price distribution shifts, but the underlying distribution of competitive positions remains fixed.

The framework yields several results that would be difficult to obtain without the decomposition. We derive pass-through formulas at each quantile of the price distribution, revealing heterogeneity invisible to single-price analysis. Consumers buying at low prices face different pass-through than those buying at high prices. Under unit demand, transaction-weighted pass-through (what matters for welfare) reduces to a sufficient statistic depending only on the consideration-set structure, so incidence can be estimated from consideration data (platform clicks, surveys, geographic proximity) rather than from a demand system, building on methods surveyed in \citet{HonkaHortacsuWildenbeest2019EmpiricalSearch}. The identification problem shifts from ``how does quantity respond to price?'' to ``who considers whom?'' which is often more tractable in this case. We derive bounds on pass-through that hold across all demand specifications, useful for policy analysis when demand is only partially known. %We also clarify when aggregate pass-through can exceed unity, showing that the same market structure yields qualitatively different incidence patterns depending on demand curvature.

We develop these results for arbitrary consideration structures before specializing to cases that admit closed-form solutions. With symmetric firms, equilibrium has a simple characterization in terms of the probability generating function for competitor counts. With asymmetric firms under independent consideration, we derive equilibria exhibiting a hierarchical structure. Higher-reach firms price higher at each quantile and, if demand is not too convex, have lower pass-through at each quantile. These results make the framework applicable to structural estimation and policy counterfactuals.

\subsection{Road map} The rest of the paper proceeds as follows. \S\ref{sec:model} sets up the paper; \S\ref{sec:isomorphism} develops the central $\mu$-isomorphism result; %showing that the pricing game is strategically equivalent to a margin game whose equilibrium depends only on the consideration structure; 
and \S\ref{sec:equilibrium} characterizes the resulting margin distributions for symmetric markets, asymmetric duopoly, and independent consideration. \S\ref{sec:quantile} then uses the isomorphism to derive quantile pass-through formulas, while \S\ref{sec:transaction} converts these posted-price results into transaction-weighted pass-through statistics. \S\ref{sec:envelopes} derives robust pass-through bounds when demand is only partially specified, and \S\ref{sec:cs} develops comparative statics linking margin orderings to price and pass-through orderings. \S\ref{sec:applications} discusses applications to gasoline markets and online retail. \S\ref{subsec:endogenous-consideration} extends the framework to endogenous consideration. %and separates conditional from general-equilibrium pass-through. \S \ref{sec:conclusion} concludes. 
All proofs omitted from the main text lie in the appendices.

\subsection{Related Literature}

This paper connects two distinct literatures: the industrial organization literature on price dispersion and limited consideration and the literature on tax incidence and pass-through.

Price dispersion in homogeneous goods markets traces to \citet{varian1980model} and \citet{BurdettJudd1983EquilibriumPriceDispersion}, who show that when consumers observe different subsets of prices, firms randomize in equilibrium.\footnote{Other classic foundations include \citet{salop_stiglitz1977_bargains_ripoffs,butters1977_equilibrium_distributions_sales_advertising,stahl1989_oligopolistic_pricing_sequential_search,janssen_moraga2004_strategic_pricing_consumer_search_number_firms}. \citet{baye_morgan2001_information_gatekeepers,baye_morgan_scholten2004_price_dispersion_small_large} diagnose dispersion with information gatekeepers and price-comparison sites.} Recent work extends these foundations in several directions: \citet{Guthmann2024PriceDispersion,Guthmann2025Unawareness, Guthmann2025} analyze dynamic settings, %\citet{elliott2021market} study platform design, 
\citet{BergemannBrooksMorris2021Search} and \citet{Albrecht2020ConsumerData} derive robust equilibrium bounds that hold across information structures, while \citet{armstrong2017ordered}, \citet{RhodesZhou2019RetailStructure}, and \citet{RhodesWatanabeZhou2021Intermediaries} explore richer consideration structures including ordered search and multiproduct settings.

Most closely connected is \citet{ArmstrongVickers2022Patterns}, who develop a general framework for competition with arbitrary consideration patterns. We adopt their consideration-set framework and equilibrium concept. They note that their equilibrium analysis is unaffected by downward-sloping demand, provided revenue is increasing up to the monopoly price (their footnote~7). \citet{McAfee1994Availability} makes a similar observation (his Remark~1). Neither paper formalizes this invariance. We provide the first formal statement and proof via the $\mu$-isomorphism, and show that the resulting decomposition yields tractable pass-through formulas. Whereas Armstrong and Vickers focus on welfare and market structure with unit demand, we characterize pass-through with standard downward-sloping demand, deriving formulas that apply to any consideration structure.

A growing empirical literature estimates consideration sets and documents their importance for demand analysis. Different subfields use different terminology: ``awareness sets'' when the friction is ignorance \citep{HonkaHortacsuVitorino2017}, ``choice sets'' when it is institutional \citep{GaynorPropperSeiler2016}, and ``attention'' when it is cognitive \citep{AbaluckAdams2021}. The core insight is the same: ignoring consideration sets biases elasticity estimates and distorts counterfactual predictions. \citet{Goeree2008LimitedInfo} shows that full-information models underestimate PC markups by a factor of four. The gap arises because high-share firms have larger consideration sets, not just better products. \citet{AbaluckAdams2021} demonstrate that consideration-set frictions explain why Medicare defaults are sticky---a pattern full-information models cannot match.

The pass-through literature, following \citet{WeylFabinger2013PassThroughTool}, provides a unified framework showing that pass-through depends on the curvature of demand relative to its slope.\footnote{\citet{bulow_pfleiderer1983_note_effect_cost_changes_prices,seade1985_profitable_cost_increases_shifting_taxation,delipalla_keen1992_ad_valorem_specific_taxation,anderson_depalma_kreider2001_tax_incidence_differentiated_oligopoly} provide earlier analyses of incidence under imperfect competition.} Their approach has been extended to vertical markets \citep{AdachiEbina2014VerticalPassThrough}, welfare analysis \citep{AdachiFabinger2022PassThrough}, multiproduct firms \citep{armstrong2023multiproduct}, and platform fees \citep{Weyl2010Platforms}.\footnote{Other canonical two-sided platform pricing frameworks are \citet{rochet_tirole2003_platform_competition_two_sided_markets,rochet_tirole2006_two_sided_markets_progress_report,armstrong2006_competition_two_sided_markets}. \citet{rysman2009_economics_two_sided_markets} surveys this literature. \citet{wang_wright2020_search_platforms_showrooming_price_parity_clauses,wang_wright2025_regulating_platform_fees} explore search-platform environments with marketplace commissions with a view toward regulation.} Empirical work documents heterogeneity across markets \citep{MarionMuehlegger2011FuelTaxIncidence,Stolper2017WhoBearsEnergyTaxes,MillerOsborneSheu2017CementPassThrough}. We extend the Weyl-Fabinger framework to markets with equilibrium price \emph{dispersion}. The same demand-curvature fundamentals matter, but they interact with a distribution of markups rather than a single markup. Our separation into competition and curvature layers parallels their decomposition into conduct and curvature, with consideration structure playing the role of conduct, but microfounded through consideration structure rather than a reduced-form parameter. 

Three recent papers have begun bridging these literatures by analyzing pass-through when prices are dispersed. \citet{GarrodLiRussoWilson2024} provide a theoretical analysis of the Varian model with general demand, showing that whether captive or non-captive consumers bear more of a cost increase depends on whether demand is log-concave or log-convex. \citet{MontagMamrakSagimuldinaSchnitzer2023} develop a Varian-style model with informed and uninformed consumers and test it using tax changes in German and French fuel markets. They find that pass-through is higher for the minimum price paid by informed consumers than for the average price paid by uninformed consumers. 
\citet{FischerMartinSchmidtDengler2024} estimate a structural search model using German fuel data, finding that informed consumers face higher effective pass-through rates and that excise tax reductions would benefit consumers more than equivalent VAT cuts. The first two papers work with binary consumer types---captive versus non-captive, or informed versus uninformed---while \citet{FischerMartinSchmidtDengler2024} estimate a distribution of search intensity.\footnote{Recently, \citet{menzio2023_drs_wp,menzio2024_trade_wp} presents compelling explorations of search-theoretic imperfect competition and markup distributions.} %We build on their insights using the Armstrong-Vickers framework, which allows arbitrary consideration patterns, and we characterize pass-through at each quantile of the price distribution rather than by consumer type.

Beyond specific applications, our approach to transaction-weighting in mixed-strategy equilibria provides a general framework for computing consumer-relevant statistics when firms use mixed strategies. This paper also connects to the literature on order statistics and quantile methods in economics \citep{koenker2005quantile, chernozhukov2013inference}. We characterize pass-through at each quantile of the price distribution, recovering heterogeneity that single-price models collapse.

\section{Model} \label{sec:model}

There are $n$ firms, indexed by $N = \{1,\ldots,n\}$. The demand side consists of a unit mass of consumers, partitioned by their consideration sets.\footnote{Other studies, such as \citet{Perla2023Awareness}, \citet{Guthmann2024PriceDispersion}, \citet{McAfee1994Availability}, \citet{Albrecht2020ConsumerData}, \citet{Guthmann2025}, and \citet{ArmstrongVickers2022Patterns}, use terms such as ``awareness,'' ``availability rate,'' ``choice set,'' ``loyal customers,'' and ``consideration set'' to indicate the subset of firms that buyers have access to.}

\begin{definition}
A \emph{consideration structure} is a probability distribution $\{\alpha_S\}_{S \subseteq N}$ where $\alpha_S \geq 0$ represents the mass of consumers who consider exactly the set $S$ of firms, with $\sum_{S \subseteq N} \alpha_S = 1$.\footnote{We allow $\alpha_\emptyset \geq 0$, representing consumers who consider no firms.}% (outside option).
\end{definition}

Each consumer observes prices only from firms in her consideration set and purchases from the lowest-priced firm, provided that price does not exceed 1.\footnote{We interpret 1 as an upper bound on feasible prices rather than a choke price---consumers have positive demand even at $p = 1$ (see \Cref{ass:demand}).}

\begin{example}[Random Search]
If each consumer samples each firm independently with probability $\lambda \in (0,1)$, then \(\alpha_S = \lambda^{|S|}(1-\lambda)^{n-|S|}\), so consideration is generated by i.i.d. Bernoulli sampling and \(\left|S\right|\) is binomially distributed.
\end{example}

\begin{example}[Spatial Markets]
Consider \(n\) equally-spaced firms located on a circle. If consumers are distributed uniformly around the circle and observe only their \(1 \leq k < n\) nearest neighbors, then \(\alpha_S = 1/n\) if \(S\) consists of \(k\) consecutive firms and \(0\) otherwise.
\end{example}

For analytical convenience, we define:
\begin{definition}
Firm $i$'s \emph{reach} is $\sigma_i \equiv \sum_{S \ni i} \alpha_S$, the mass of consumers who consider $i$.\footnote{Throughout, we restrict attention without loss of generality to firms with strictly positive reach.} Firm $i$'s \emph{captive share} is $\alpha_{\{i\}}$, the mass who consider only $i$. Firm $i$'s \emph{captive-to-reach ratio} is $\rho_i \equiv \alpha_{\{i\}}/\sigma_i \in [0,1]$.
\end{definition}
Each consumer who purchases at price $p$ demands quantity $x(p)$ where:

\begin{assumption}\label[assumption]{ass:demand}
The function $x \colon (0,1] \to \R_+$ is continuous, weakly decreasing, and continuously differentiable, with $x(1) > 0$.
\end{assumption}

Firms share a common marginal cost $c \in [0,1)$. The cost may represent production costs, taxes, or input prices whose changes we wish to study. Firms simultaneously choose price distributions. We break ties uniformly: whenever multiple firms offer the same lowest price within a consumer's consideration set, we assume that the consumer randomizes uniformly across the tied lowest-priced firms. Formally, for any nonempty consideration set \(S\subseteq N\) and realized price profile
\(\mathbf{p} \in[c,1]^N\), define the set of minimizers
\[
M_S(\mathbf{p})\coloneqq \arg\min_{j\in S} p_j.
\]

A consumer with consideration set \(S\) purchases from firm \(i\in S\) with probability
\[
\frac{\mathbf{1}\left\{i\in M_S(\mathbf{p})\right\}}{\left|M_S(\mathbf{p})\right|}.
\]
Consequently, given rival mixed strategies \(F_{-i}\), the \emph{demand} faced by firm \(i\) when it posts price \(p\) is, therefore,
\[\label{eq:demand}\tag{Demand}
q_i(p)\coloneqq \sum_{S\ni i}\alpha_S\cdot
\mathbb{E}\left[
\frac{\mathbf{1}\left\{i\in M_S\left(p, p_{-i}\right)\right\}}
{\left|M_S\left(p, p_{-i}\right)\right|}
\right].
\]

If the mixed-strategy cumulative distribution functions (CDFs) are atomless on the interior of their support (which is typical in price competition with continuous payoffs on each side of any price), so ties occur with probability zero,
the demand formula simplifies to
\[\label{eq:demand1}\tag{Demand\(^*\)}
q_i(p) = \sum_{S \ni i} \alpha_S \prod_{j \in S \setminus \{i\}} (1 - F_j(p))\text{,}\]
where we follow the convention that the empty product (when $S = \{i\}$) equals \(1\). \eqref{eq:demand1} is the mass of consumers who consider $i$ and find $i$ strictly cheaper than all rivals in their consideration set.

Firm $i$'s profit from posting price $p$ is \[\label{eq:profit}\tag{Profit}
\Pi_i(p;c) = (p-c)x(p)q_i(p)\] A profile of price CDFs $(F_i)_{i \in N}$ on $[c,1]$ constitutes an equilibrium if for each firm $i$ two conditions hold:
\begin{enumerate}
    \item Each firm is indifferent over all prices on the support of its distributions: for all prices $p$ in the support of $F_i$, $\Pi_i(p;c) = \pi_i(c)$.
    \item No firm has a profitable deviation: for all prices $p \in [c,1]$ outside the support of $F_i$, $\Pi_i(p;c) \leq \pi_i(c)$.
\end{enumerate}

We also make the following monotonicity assumption on the effective margin per served consumer:
\begin{assumption}\label[assumption]{ass:invertible}
The \emph{effective margin per served consumer}, function $(p-c)x(p)$, has strictly positive derivative on $[c,1]$: $x(p) + (p-c)x'(p) > 0$ for all $p \in [c,1]$.
\end{assumption}
This assumption ensures that higher prices correspond to higher effective margins, maintaining the trade-off between margin and market share.  For the standard demand classes considered below, this assumption is equivalent to the effective margin $(p-c)x(p)$ being maximized at the boundary $p = 1$ rather than at an interior price. To ensure invertibility, we will always need some assumption like this. With unit demand, which is common in the search and price dispersion literature, the assumption holds automatically since $(p-c) \cdot 1$ is linear in $p$. 

When the effective margin is instead maximized at an interior price $\hat{p}(c) < 1$, the assumption fails on $[c,1]$. For example, with CES demand \(x(p)=p^{-\eta}\), the effective margin is \((p-c)p^{-\eta}\), with derivative \(p^{-\eta-1}[(1-\eta)p+\eta c]\). Consequently, \Cref{ass:invertible} holds automatically for \(\eta\in[0,1)\), holds for \(\eta=1\) whenever \(c>0\) (while \(c=0\) makes \(p\mapsto (p-c)x(p)\)
constant on \((0,1]\)), and for \(\eta>1\) holds on \([c,1]\) if and only if \(c> 1-1/\eta\). When \(\eta>1\) and \(c\in(0,1-1/\eta]\), the effective margin is maximized at \(\hat p(c)=\eta c/(\eta-1)\in(c,1]\). The isomorphism can be extended in this case by normalizing margins by
\(\bar m(c)\equiv \max_{p\in[c,1]}(p-c)x(p)\) rather than \((1-c)x(1)\), and restricting attention to prices \(p\le \hat p(c)\), since any price above \(\hat p(c)\) yields the same effective margin as a lower price and a weakly lower demand share. The competition layer is unchanged---what changes is the curvature layer, since \(\bar m'(c)=-x(\hat p(c))\) when the maximum is interior. We maintain \Cref{ass:invertible} throughout to preserve the simpler structure. %For example, with CES demand $x(p) = p^{-\eta}$, the effective margin is maximized at $\hat{p}(c) = \eta c/(\eta-1)$, so the assumption holds if and only if $c \geq 1 - 1/\eta$. The isomorphism can be extended to this case by normalizing margins by $\bar{m}(c) \equiv \max_{p \in [c,1]}(p-c)x(p)$ rather than $(1-c)x(1)$. The competition layer is unchanged; what changes is the curvature layer, since $\bar{m}'(c) = -x(\hat{p}(c))$ varies with $c$ when the maximum is interior. We maintain \Cref{ass:invertible} throughout to preserve the simpler structure.

\section{The \texorpdfstring{\(\mu\)-Isomorphism}{isomorphism}} \label{sec:isomorphism}

This section establishes our central theoretical result: the pricing game with standard demand is isomorphic to a unit-demand game after a change of variables. Prior work has noted this invariance informally. \citet{ArmstrongVickers2022Patterns} observe in their footnote~7 that their analysis ``is not affected if each consumer has a downward-sloping demand function $x(p)$, provided revenue $px(p)$ is an increasing function up to the monopoly price,'' and \citet{McAfee1994Availability} notes that ``there is no difference in profits between different demand curves with the same monopoly profits'' (Remark~1). Neither paper formalizes this observation. We do so here via a change of variables that makes the invariance precise and yields a decomposition useful for pass-through analysis. Firms compete over effective margin per customer served---how much profit they extract from each transaction. Once we normalize these margins appropriately, the strategic problem becomes identical regardless of the underlying demand function.

\subsection{Normalized Effective Margins}

To see the isomorphism, start with what firms compete over. When a firm sets price $p$, it earns $(p-c)x(p)$ from each customer it serves. This is the effective margin. Different demand functions $x(\cdot)$ change the mapping from prices to effective margins, but the strategic trade-off is the same. Higher margins mean higher profit per customer but lower probability of winning customers.

\begin{definition}The \emph{normalized effective margin} at price $p$ is:
\[\label{eq:mu_def}\tag{Normalized Effective Margin}
\mu(p;c) \equiv \frac{(p-c)x(p)}{(1-c)x(1)} \in [0,1].\]
\end{definition}%Given demand $x(\cdot)$ and cost $c$

The numerator is firm profit per served consumer. The denominator is the profit per served consumer when pricing at the maximum price. This normalization maps all possible effective margins to the unit interval $[0,1]$, with $\mu = 0$ corresponding to pricing at cost (zero margin) and $\mu = 1$ corresponding to extracting maximum profit via the maximum price. Under \Cref{ass:invertible}, $\mu(p;c)$ is strictly increasing in \(p\) on the equilibrium support, ensuring a bijection between prices and normalized margins.

The inverse map \(\phi(\mu,c)\) solves:
\[\label{eq:phi_def}\tag{Inverse Map}
(\phi(\mu,c) - c)x(\phi(\mu,c)) = \mu(1-c)x(1).\]
We record the existence and uniqueness of the inverse map \(\phi\):
\begin{lemma}\label[lemma]{lem:phi_exists}
Under \Cref{ass:demand,ass:invertible}, for each $\mu \in [0,1]$ and $c \in [0,1)$, there exists a unique inverse map $\phi(\mu,c) \in [c,1]$ satisfying \eqref{eq:phi_def}.
\end{lemma}

\subsection{The Main Isomorphism Result}

Once we reformulate the game in terms of normalized effective margins $\mu$, the demand-curvature dependence vanishes. Equilibrium depends only on the consideration structure. Demand curvature and costs matter only for translating these margin distributions back into price distributions.

We define the margin game firms take part in as follows.

\begin{definition}\label{def:margin-game}
In the \emph{margin game}, firms simultaneously choose probability distributions over margins $\mu \in [0,1]$. When firm $i$ posts margin $\mu$ and rivals use CDFs $(F_j^\mu)_{j\neq i}$, firm $i$'s payoff is $\Pi_i^\mu(\mu) = \mu \cdot q_i^\mu(\mu)$, where the demand share $q_i^\mu(\mu)$ is defined analogously to \eqref{eq:demand}:
\[q_i^\mu(\mu) = \sum_{S\ni i}\alpha_S\cdot
\mathbb{E}\left[
\frac{\mathbf{1}\left\{i\in M_S\left(\mu, \mu_{-i}\right)\right\}}
{\left|M_S\left(\mu, \mu_{-i}\right)\right|}
\right],\]
with $M_S(\mu, \mu_{-i})$ denoting the set of margin-minimizers in $S$ and ties broken uniformly. When rival CDFs are atomless at $\mu$, ties occur with probability zero and the demand share simplifies  as in \eqref{eq:demand1}:
\[q_i^\mu(\mu) = \sum_{S\ni i} \alpha_S \prod_{j\in S\setminus\{i\}} \left[1-F_j^\mu(\mu)\right].\]
\end{definition}

An equilibrium in the margin game exists:
\begin{lemma}\label{lem:existence}
The margin game admits a mixed-strategy Nash equilibrium.
\end{lemma}
The proof verifies the hypotheses of \citet[Theorem~5]{DasguptaMaskin1986DiscontinuousGames}.%, appears in Appendix~\ref{app:existence-proof}.

\begin{theorem}\label{thm:isomorphism}
Consider the pricing game with demand $x(\cdot)$ satisfying \Cref{ass:demand,ass:invertible} and cost $c$. Then:
\begin{enumerate}[label=(\alph*),nosep]
\item\label{item:bijection} The map $\Phi$ defined by $\Phi(F_i)(\mu) \coloneqq F_i(\phi(\mu,c))$, so that $F_i^\mu = \Phi(F_i)$, is a bijection between equilibria of the pricing game and equilibria of the margin game.\footnote{Thus, the pricing game admits a mixed-strategy equilibrium.}
\item\label{item:invariance} The equilibrium $\mu$-distributions depend only on the consideration structure $(\alpha_S)_{S\subseteq N}$, not on $x(\cdot)$ or \(c\).
\end{enumerate}
\end{theorem}
The normalized effective margin $\mu(p;c)$ gives a bijection between prices and margins, and payoffs in the pricing game are proportional to payoffs in the margin game with a common scaling factor $(1-c)x(1)>0$, \textit{which does not affect best-response comparisons}.

\begin{remark}
The invertibility condition that $p\mapsto (p-c)x(p)$ is strictly increasing ensures that the effective margin is strictly increasing in price. Without this, multiple prices could yield the same margin, breaking the bijection. Economically, the condition requires that higher prices always generate higher profit per customer served; \textit{viz.}, firms never face a ``backward-bending'' margin curve. This holds for unit demand (where $x(p)\equiv 1$) and for most standard demand specifications when markups are moderate.
\end{remark}

The $\mu$-isomorphism clarifies pass-through analysis. Pass-through decomposes into two distinct layers: the competition layer solves for equilibrium $\mu$-distributions using only the consideration structure $\{\alpha_S\}$, determining how market power is distributed across the price distribution. The curvature layer maps normalized margins to prices via $p = \phi(\mu,c)$ using the demand function $x(\cdot)$ and cost $c$, determining how a given level of market power translates into actual prices. Pass-through then follows by differentiating the mapping: $\tau = \phi_c(\mu,c)$. The pass-through rate depends on both layers but in a separable way.

This separation simplifies the analysis. Rather than solving different equilibria for each $(c, x(\cdot))$ combination, we solve once in $\mu$-space and apply different transformations. Note that the equilibrium $\mu$-distribution is invariant to cost changes: when costs rise, the entire price distribution shifts, but the underlying distribution of market power (captured by $\mu$) remains fixed. When multiple equilibria exist, this invariance holds for the set.%, but the uniqueness results in \S\ref{sec:equilibrium} ensure a well-defined equilibrium path for differentiation.

\subsection{Economic Intuition} The normalized effective margin $\mu$ captures the trade-off in price-setting. Higher $\mu$ means higher profit per customer but lower probability of winning customers. Think of $\mu$ as the firm's ``aggressiveness'' in extracting surplus. Choosing $\mu$ is like choosing a position on the competition spectrum, from aggressive (low $\mu$, low margins, high market share) to passive (high $\mu$, high margins, low market share).

The consideration structure $\{\alpha_S\}$ determines how this trade-off resolves in equilibrium. Markets with more overlapping consideration sets push firms toward lower $\mu$, while markets with many captive customers allow higher $\mu$. Crucially, the demand curvature $x(\cdot)$ \emph{only affects the translation between these strategic positions and actual prices, not the positions themselves}.

\subsection{Application to Merger Analysis}
The $\mu$-isomorphism yields a decomposition for merger analysis. Any merger changes the consideration structure ${\alpha_S}$, inducing new equilibria in the margin game. Since the margin game is invariant to demand, this step can be computed once without specifying demand. The price effects at each quantile then follow from $\phi(\mu^{\mathrm{post}}(u), c) - \phi(\mu^{\mathrm{pre}}(u), c)$. This separates the problem: the competition layer determines how margins shift, whereas the curvature layer determines how those margin shifts translate to prices. The first depends only on consideration structure, the second only on demand.\footnote{For a tractable multiproduct oligopoly approach via aggregative games (with merger applications), see, in particular, \citet{nocke_schutz2018_multiproduct_firm_oligopoly_aggregative_games, nocke_schutz2025_aggregative_games_merger_analysis}. For complementary merger-approximation tools, see \citet{jaffe_weyl2013_first_order_merger_analysis}.}

\section{Equilibrium Characterization} \label{sec:equilibrium}

Having established the isomorphism, we now characterize equilibrium in $\mu$-space. The cases we consider---symmetric firms, asymmetric duopoly, and independent consideration---all fall within the ``symmetric interactions'' class of \citet{ArmstrongVickers2022Patterns} (which requires symmetry of the consideration structure's interaction parameters, not of firms' reaches), who prove existence, uniqueness, and the structural properties of equilibrium (nested supports, common minimum price, profits proportional to reach). We reformulate their equilibria in terms of normalized margins $\mu$ and derive the explicit quantile functions used below.

\subsection{Symmetric Firms}

When the consideration structure treats all firms symmetrically (for instance, when consumers randomly sample firms or when firms are arranged symmetrically in geographic or product space), equilibrium has a simple form: symmetric competition leads to a common distribution of market power.%, though firms still mix over different prices in equilibrium.

\begin{definition}\label{def:symmconsider}
The structure $\{\alpha_S\}$ is \emph{symmetric} if $\alpha_S$ depends only on $|S|$, not on the identity of firms in $S$.
\end{definition}

Under symmetry, all firms have identical reach $\sigma$ (the mass of consumers who consider them) and captive-to-reach ratio $\rho = \alpha_{\{i\}}/\sigma$ (the fraction of their potential customers who consider no other firms). The ratio $\rho$ measures the degree of market power arising from limited consideration. When $\rho$ is high, many consumers are captive to individual firms; when $\rho$ is low, most consumers compare multiple options.

\begin{proposition} \label[proposition]{prop:symmetric}
With symmetric consideration structure, the unique symmetric equilibrium in $\mu$-space has quantile function \(\mu(u) = \frac{\rho}{H(1-u)}\), where $H(s) = \frac{1}{\sigma}\sum_{S \ni i} \alpha_S s^{|S|-1}$ is the probability generating function of $|S|-1$ conditional on $i \in S$.
\end{proposition}

$H(s)$ encodes the competitive environment: it tells us the distribution of how many rival firms a consumer considers, conditional on considering firm $i$. The formula shows that firms mix over higher margins (higher $\mu$) when they have more captive customers (higher $\rho$). Perhaps less obviously, margins are also higher when consumers consider more rivals (lower $H$). The logic, which goes back to \citet{Rosenthal1980}, is that with more rivals, competing aggressively for contested consumers becomes less valuable since they are spread across more firms. Firms respond by shifting their mixing toward higher margins, focusing on extracting surplus from captive customers rather than competing fiercely for contested ones.

\subsection{Example: Binomial Consideration}

To make these concepts concrete, consider a market where each consumer independently considers each firm with probability $\lambda$. This captures settings like online markets where consumers randomly encounter products, or markets where advertising reaches consumers stochastically.

In this case, firm \(i\)'s reach is $\sigma_i = \lambda$ and captive mass is $\alpha_{\{i\}} = \lambda(1-\lambda)^{n-1}$. The captive-to-reach ratio is $\rho = \alpha_{\{i\}}/\sigma_i = (1-\lambda)^{n-1}$, and the PGF of $|S|-1$ conditional on $i \in S$ is $H(s) = (\lambda s + (1-\lambda))^{n-1}$. The equilibrium $\mu$-quantile function is:\[\mu(u) = \frac{\rho}{H(1-u)} = \frac{(1-\lambda)^{n-1}}{(\lambda(1-u) + (1-\lambda))^{n-1}} = \left(\frac{1-\lambda}{1 - \lambda u}\right)^{n-1}.\]

\subsection{Asymmetric Duopoly}

Any duopoly has symmetric interactions, so Proposition~1 of \citet{ArmstrongVickers2022Patterns} guarantees a unique equilibrium with common lower support bound and profits proportional to reach. We derive the explicit $\mu$-space CDFs.

Consider $n = 2$ firms with captive shares $\alpha_1 \equiv \alpha_{\{1\}}$ and $\alpha_2 \equiv \alpha_{\{2\}}$, shared segment $\alpha_{12} \equiv \alpha_{\{1,2\}}$ (consumers who consider both), and outside option $\alpha_\emptyset$. Define the captive-to-reach ratios:
\[
\rho_1 = \frac{\alpha_1}{\alpha_1 + \alpha_{12}}, \quad \text{and} \quad \rho_2 = \frac{\alpha_2}{\alpha_2 + \alpha_{12}}
\]
The captive-to-reach ratio $\rho_i$ measures the fraction of firm $i$'s potential customers who have no alternative. Higher $\rho_i$ means more market power from captive consumers.

\begin{proposition}\label[proposition]{prop:asymmetric}
Assume $1>\rho_i >  0$. Let $\rho_1 > \rho_2$. The unique equilibrium in $\mu$-space has both firms mixing on the common support $[\underline{\mu}, 1]$ with lower bound $\underline{\mu} = \rho_1$. The equilibrium CDFs on $\mu \in [\underline{\mu}, 1)$ are\footnote{Firm 1 (with higher $\rho_1$) has a mass point at $\mu = 1$ of size $\Delta_1 = 1 - F_1^\mu(1^-) = (\rho_1 - \rho_2)/(1 - \rho_2)$, while firm 2 has no atom ($F_2^\mu(1) = 1$).}
\[F_1^\mu(\mu) = 1 - \frac{1}{1-\rho_2} \left( \frac{\rho_1}{\mu} - \rho_2 \right) \quad \text{and} \quad
F_2^\mu(\mu) = 1 - \frac{\rho_1}{1-\rho_1} \left(\frac{1-\mu}{\mu}\right).\]
\end{proposition}
The corresponding asymmetric duopoly quantile functions are
\[\mu_1(u) = \begin{cases}
\frac{\rho_1}{1 - u(1-\rho_2)} \quad &\text{if} \quad u \leq 1 - \Delta_1 \\
1 \quad &\text{if} \quad u > 1 - \Delta_1,
\end{cases} \quad \text{and} \quad
\mu_2(u) = \frac{\rho_1}{1 - u(1-\rho_1)}.\]

The asymmetric equilibrium reveals how differences in captive shares shape competitive behavior. The firm with more captives ($\rho_1 > \rho_2$) prices at the monopoly margin $\mu = 1$ with positive probability. The firm with fewer captives earns rents: it benefits from the ``price umbrella'' set by the stronger firm, earning expected profits strictly above its captive value ($\pi_2^* = \rho_1(\alpha_2 + \alpha_{12}) > \alpha_2$). Both firms share the same support-lower-bound $\underline{\mu} = \rho_1$, determined by the stronger firm's desire to exploit its large captive base.

We can also order the duopoly margins. For $u \leq 1 - \Delta_1$, the quantile functions satisfy:
\[
\mu_1(u) = \frac{\rho_1}{1 - u(1-\rho_2)} > \frac{\rho_1}{1 - u(1-\rho_1)} = \mu_2(u),
\]
since $\rho_1 > \rho_2$ implies $1-\rho_1 < 1-\rho_2$. Firm 1 (with more captive customers) maintains higher margins and prices at every quantile.

\subsection{Extension to n Asymmetric Firms: Independent Consideration}

The duopoly analysis extends to $n$ asymmetric firms when consideration sets exhibit a particular structure: \emph{independent awareness}. Under this assumption, whether a consumer considers firm $i$ is statistically independent of whether she considers firm $j$. This structure, used first by \citet{Ireland1993} and \citet{McAfee1994Availability}, yields closed-form equilibrium characterizations for arbitrary $n$.

\begin{assumption}\label[assumption]{ass:independence}
Each consumer considers firm $j$ independently with probability $\lambda_j \in (0,1)$. The consideration structure is
\(\alpha_S = \prod_{j \in S} \lambda_j \prod_{k \notin S}(1-\lambda_k)\).
\end{assumption}

Under independence, firm $j$'s reach is $\sigma_j = \lambda_j$ and its captive share is $\alpha_{\{j\}} = \lambda_j \prod_{k \neq j}(1-\lambda_k)$. The captive-to-reach ratio becomes:
\[
\rho_j = \frac{\alpha_{\{j\}}}{\sigma_j} = \prod_{k \neq j}(1-\lambda_k).
\]

A key property of independence is that firm $j$'s demand share when posting $\mu$ takes a multiplicatively separable form:
\[\label{eq:indep_demand}\tag{\(\mathrm{Demand}_I\)}
q_j^\mu(\mu) = \sum_{S \ni j} \alpha_S \prod_{i \in S \setminus \{j\}}(1 - F_i^\mu(\mu)) = \lambda_j \prod_{i \neq j}\left[1 - \lambda_i F_i^\mu(\mu)\right].\]
This separability is what enables closed-form solutions. In particular, we state the $\mu$-space equilibrium and verify it directly. For the remainder of this section we posit \Cref{ass:independence} and order firms so that \(\lambda_1>\lambda_2 > \lambda_3 \ge \cdots \ge \lambda_n\), where we impose the first two strict inequalities to simplify the equilibrium construction. With this arrangement in hand, we define the common lower bound \(\underline{\mu}\coloneqq \rho_1=\prod_{h=2}^n(1-\lambda_h)\) and the upper support bounds \(\bar{\mu}_1=\bar{\mu}_2\coloneqq 1\), and for each \(k\in\{3,\dots,n\}\),
\[
\bar{\mu}_k \coloneqq \frac{\prod_{h=2}^{k-1}(1-\lambda_h)}{(1-\lambda_k)^{k-2}} \in (0,1), \quad \text{with the convention } \bar{\mu}_{n+1}\coloneqq \underline{\mu}.
\]
For each \(m\in\{2,\dots,n\}\), we also define $C_m\coloneqq \prod_{h=m+1}^n(1-\lambda_h)$ (so $C_n=1$ and $C_2=\prod_{h=3}^n(1-\lambda_h)$), and for \(\mu\in[\bar{\mu}_{m+1},\bar{\mu}_m]\) define the common multiplier $\Gamma(\mu)\coloneqq 1-\left(\underline{\mu}/(\mu\,C_m)\right)^{1/(m-1)}$.
\begin{proposition}\label[proposition]{prop:n_asymmetric}
The unique equilibrium in \(\mu\)-space has the following structure:
\begin{enumerate}
    \item Firm 1 mixes on \(\left[\underline{\mu},1\right]\) and has an atom at \(\mu=1\). Firm 2 mixes continuously on \(\left[\underline{\mu},1\right]\) with no atom. Each firm \(k\ge 3\) mixes continuously on \([\underline{\mu},\bar{\mu}_k]\) with \(\bar{\mu}_k<1\). The supports are nested: $1=\bar{\mu}_1=\bar{\mu}_2>\bar{\mu}_3\ge \cdots \ge \bar{\mu}_n>\underline{\mu}$.
    \item\label{it:452} Each firm \(j\) earns equilibrium profit $\pi_j^\ast=\lambda_j\,\underline{\mu}$, with \(\pi_1^\ast=\lambda_1\underline{\mu}=\alpha_{\{1\}}\).% in particular.
    \item For each \(m\in\{2,\dots,n\}\), for each \(\mu\in[\bar{\mu}_{m+1},\bar{\mu}_m]\),
\[
F_j^\mu(\mu)=
\begin{cases}
\Gamma(\mu)/\lambda_j, & j\le m,\\
1, & j>m,
\end{cases}
\qquad\text{and}\qquad
F_j^\mu(\mu)=0\ \text{ for }\mu<\underline{\mu}.
\]
On the common overlap \([\underline{\mu},\bar{\mu}_n]\) (where \(m=n\)),\footnote{Firm 1 has a mass point at \(\mu=1\) of size $\Delta_1 = 1-F_1^\mu(1^-)=1-\frac{\lambda_2}{\lambda_1}$. No other firm has an atom there.} %All firms \(j\ge 2\) have no atom at \(\mu=1\).}
\[\label{eq:indep_cdf}\tag{\(1\)}
\Gamma(\mu)=1-\left(\underline{\mu}/\mu\right)^{1/(n-1)} \quad\text{and}\quad F_j^\mu(\mu)=\frac{\Gamma(\mu)}{\lambda_j}.
\]
\end{enumerate}
\end{proposition}

The quantile functions follow by inverting the CDFs. %Since $F_j^\mu(\mu) = \Gamma(\mu)/\lambda_j$ on each interval where firm $j$ is active, and the active set changes at the cutoffs $\bar{\mu}_k$, the quantile function is piecewise. For firm $j \geq 2$ (no atom), on the interval where exactly $m$ firms are active---i.e., for $u \in [\lambda_{m+1}/\lambda_j,\, \lambda_m/\lambda_j]$, with $m = n, n-1, \ldots, j$ and convention $\lambda_{n+1} \coloneqq 0$---the quantile function is \[\mu_j(u) = \frac{\underline{\mu}}{C_m\left(1 - \lambda_j u\right)^{m-1}}.\] On the common overlap $[\underline{\mu}, \bar{\mu}_n]$ (where $m = n$ and $C_n = 1$), this reduces to $\mu_j(u) = \underline{\mu} \cdot (1 - \lambda_j u)^{-(n-1)}$. For firm~1, the continuous part follows the same piecewise formula on $u \in [0, \lambda_2/\lambda_1]$, with $\mu_1(u) = 1$ for $u > \lambda_2/\lambda_1$ (the atom).
We note the following stochastic dominance of margins, namely, that higher-reach firms price higher at each quantile.
\begin{corollary}\label[corollary]{cor:margin_dominance}
The margin CDFs satisfy $F_1^\mu(\mu) \leq F_2^\mu(\mu) \leq \cdots \leq F_n^\mu(\mu)$ for all $\mu$. % in the common support. 
Equivalently, $\mu_1(u) \geq \mu_2(u) \geq \cdots \geq \mu_n(u)$ for all $u \in [0,1]$.
\end{corollary}

The independent consideration model reveals a clean hierarchical structure. The key equilibrium object is \(\Gamma(\mu)=\lambda_j F_j^\mu(\mu)\), which is common across active firms on each support interval. Hence, whenever firm \(j\) is active, \(F_j^\mu(\mu)=\Gamma(\mu)/\lambda_j\). Higher-reach firms (\(\lambda_j\) large), therefore, have flatter CDFs: they place less mass on low margins and, equivalently, maintain higher margins at each quantile. The leader’s atom at \(\mu=1\) reflects this hierarchy. Firm 1 has the largest equilibrium profit target, \(\pi_1^*=\lambda_1\underline{\mu}\), and the flattest CDF, so its continuous part does not exhaust all mass below \(1\); indeed, \(F_1^\mu(1^-)=\lambda_2/\lambda_1<1\), and the remaining mass is placed at the monopoly margin.

Lower-reach firms have less captive power under independence. Within a given market,
\[
\rho_j=\prod_{k\neq j}(1-\lambda_k)=\frac{\alpha_\emptyset}{1-\lambda_j},
\]
so \(\rho_j\) is increasing in \(\lambda_j\). Instead, what lower-reach firms have are smaller equilibrium profit targets and steeper CDFs, which allow them to achieve those profits without an atom at \(\mu=1\). Firm 2 reaches \(\mu=1\) only continuously, while firms \(k\ge 3\) have supports ending at \(\bar{\mu}_k<1\). The common lower bound \(\underline{\mu}=\prod_{h\neq 1}(1-\lambda_h)=\rho_1\) depends on the reach of all rival firms: when rivals have higher reach, the competitive floor on margins falls. Pass-through heterogeneity across firms then follows from their positions in the margin distribution. Higher-reach firms post higher margins at each quantile and, when \(\phi_c\) is decreasing in \(\mu\) (as under unit or linear demand), have lower pass-through at each quantile.

Of course, the independent consideration structure is a special case of general consideration sets. The key simplification is that the demand share \eqref{eq:indep_demand} factors multiplicatively, implying that $\lambda_j F_j^\mu(\mu)$ must be equal across active firms on each support interval. \eqref{eq:indep_cdf} solves this equal-$\Gamma$ structure explicitly on the lowest interval where all $n$ firms compete. For general (correlated) consideration structures, this factorization fails. Nevertheless, for duopoly, the qualitative properties (common support, stochastic dominance ordering, atoms for high-reach firms, pass-through ranking) hold for general consideration structures.%; whether they extend to $n \geq 3$ with correlated consideration remains an open question.

\section{Quantile Pass-Through} \label{sec:quantile}

We now derive our main pass-through results. The $\mu$-isomorphism lets us compute pass-through rates by differentiating the mapping from margins to prices. All the action comes from how $\phi$ responds to costs. Throughout, fix a firm $i$ with equilibrium $\mu$-quantile function $\mu_i(u) \coloneqq \inf \left\{\mu \colon \ F_i^\mu (\mu) \geq u \right\}$.

\subsection{The Pass-Through Formula}

Quantile pass-through tells us how each price in the distribution responds to cost changes. Because firms randomize in equilibrium, different quantiles of the price distribution can have different pass-through rates. This heterogeneity in pass-through across the price distribution is a key feature of markets with price dispersion.

\begin{definition}\label{def:7}
The \emph{pass-through rate at quantile \(u\)} is \(\tau_i^Q(u;c) \equiv \frac{\partial p_i(u;c)}{\partial c}\), where $p_i(u;c) = \phi(\mu_i(u),c)$.
\end{definition}

\begin{theorem} \label{thm:quantile}
Under \Cref{ass:demand,ass:invertible}, the quantile pass-through rate is:
\[\label{eq:passthrough}\tag{\(2\)}
\tau_i^Q(u;c) = \phi_c(\mu_i(u),c) = \frac{x(p_i(u;c))(1-p_i(u;c))}{(1-c)[x(p_i(u;c)) + (p_i(u;c)-c)x'(p_i(u;c))]}
\]
\end{theorem}

\subsection{Economic Interpretation}

The pass-through formula \eqref{eq:passthrough} factors as
\[\tau_i^Q(u;c) = \frac{1-p}{(1-c)(1-\varepsilon(p))}, \qquad \text{and} \qquad \varepsilon(p) \equiv -\frac{(p-c)x'(p)}{x(p)} \in [0,1),\]
since $x(p) + (p-c)x'(p) = x(p)(1-\varepsilon)$ is the slope of the effective-margin function $(p-c)x(p)$. Note that $\varepsilon$ is a property of per-customer demand $x(\cdot)$, not of the firm's equilibrium residual demand $x(p)q_i(p)$, which is endogenous to the equilibrium.

The numerator $(1-p)/(1-c)$ is the firm's normalized headroom below the reservation value. The denominator $(1-\varepsilon)$ measures the slope of the effective margin. When $\varepsilon$ is close to 1, the effective margin is nearly flat and the firm must adjust price substantially to preserve $(p-c)x(p) = \mu(1-c)x(1)$ after a cost shock. Hence, more elastic $x(\cdot)$ at a fixed quantile raises $\tau$.

\section{Transaction-Weighted Pass-Through} \label{sec:transaction}

While quantile pass-through describes how each price in the distribution responds to costs, welfare analysis requires understanding what consumers actually pay. Consumers do not randomly draw from the price distribution but rather systematically buy more at lower prices.\footnote{Empirically, consumers disproportionately purchase at low prices, so dispersion in posted prices need not mirror dispersion in transaction prices. Reference, e.g., \citet{Sorensen2000Drugs,kaplan_menzio_rudanko_trachter2019_relative_price_dispersion}.} This selection effect fundamentally changes how we think about pass-through and incidence.

A key simplification comes from the equilibrium indifference condition: the mass of transactions at price $p$ satisfies $T_i(p;c) = \pi_i(c)/(p-c)$, so transaction volume is proportional to $1/(p-c)$, independent of the demand function $x(\cdot)$. Firms selling at lower prices must compensate with proportionally higher volume to maintain the same profit. Define the \emph{harmonic integral} $B_i(c) \equiv \int (p-c)^{-1} dF_i(p;c)$, in which case the mean transaction-weighted price is then $\bar{p}_i^{\trans}(c) = c + 1/B_i(c)$.

\begin{proposition}\label[proposition]{prop:transaction}
Posit \Cref{ass:demand,ass:invertible}. If a firm \(i\)'s posted-price support is bounded away from $c$, the transaction-weighted pass-through rate is
\[\label{eq:trans_passthrough}\tag{\(3\)}
\tau_i^{\trans}(c) = 1 + \frac{\int_0^1 \frac{\phi_c(\mu_i(u),c) - 1}{(\phi(\mu_i(u),c) - c)^2} du}{\left[\int_0^1 \frac{1}{\phi(\mu_i(u),c) - c} du\right]^2}.\]
\end{proposition}

Under unit demand, this reduces to a closed-form expression in terms of the consideration structure alone. For each firm \(i\), with $\mu_i(u)$ being firm $i$'s equilibrium quantile function in the margin game, we define \(K_i \coloneqq \int_0^1 \frac{1}{\mu_i(u)} du\). Under unit demand, \(B_i(c) = \frac{K_i}{1-c}\) for each firm \(i\). Thus,

\begin{corollary} \label[corollary]{cor:trans_asymmetric}
Under unit demand ($x(p)=1$), firm $i$'s mean transaction price and transaction-weighted pass-through satisfy
\[\bar p_i^{\trans}(c) = c + \frac{1-c}{K_i},
\qquad \text{and} \qquad
\tau_i^{\trans} \equiv \frac{d\bar p_i^{\trans}(c)}{dc}
= 1 - \frac{1}{K_i}.
\label{eq:trans_unit_demand_general}\tag{\(4\)}\]
\end{corollary}

The object $K_i$ measures the intensity of competition facing firm $i$, aggregated across its transaction distribution. When the margin distribution places substantial weight on low $\mu$ (contested transactions), $K_i$ is large and pass-through is high: firms facing stiff competition pass cost shocks to consumers. When the distribution concentrates on high $\mu$ (captive transactions), $K_i$ is small and pass-through is low: firms with market power absorb cost shocks.

Under unit demand and symmetric consideration, the equilibrium quantile function satisfies $\mu(u)=\rho/H(1-u)$, so \(K = \bar H/\rho\), where \(\bar H \equiv \int_0^1 H(s)ds\), and the general formula \eqref{eq:trans_unit_demand_general} reduces to \(\tau^{\trans} = 1 - \rho/\bar H\).

Under independent consideration, the per-firm formulas aggregate nicely. %to a closed-form industry-level expression. 
Let $\Lambda \equiv \sum_{k=1}^{n} \lambda_k$ denote total reach and $\alpha_\emptyset \equiv \prod_{k=1}^{n}(1-\lambda_k)$ the no-purchase mass.

\begin{proposition}
\label[proposition]{prop:aggregate_margin}
Posit \Cref{ass:independence} and order firms so that $\lambda_1 > \lambda_2 \geq \cdots \geq \lambda_n$. The aggregate expected transaction margin is
\[\label{eq:aggregate_margin}\tag{\(5\)}
\bar{\mu}^{\trans}
=
\frac{\underline{\mu}\,\Lambda}{1-(1-\lambda_1)\underline{\mu}}.\]
\end{proposition}

\begin{corollary}
\label[corollary]{cor:industry_passthrough}
Under unit demand and independent consideration, the industry pass-through rate is
\[\tag{\(6\)}\label{eq:industry_passthrough}
\tau^{\trans, \text{agg}}
=
1 - \frac{\underline{\mu}\,\Lambda}{1-(1-\lambda_1)\underline{\mu}}.\]
\end{corollary}

Industry pass-through is one minus the average market power among transacting consumers. The aggregate margin is determined by three scalars: the support floor $\underline{\mu}$, the total reach $\Lambda$, and the leader's reach $\lambda_1$.

There is a natural application to mergers. Standard merger analysis focuses on price levels. But mergers also affect how future cost shocks are transmitted to consumers. By the separation principle, a merger changes the consideration structure, which shifts each firm's equilibrium margin distribution. Under unit demand, the new $K_i = \int \mu_i(u)^{-1} du$ determines $\tau_i^{\mathrm{trans}} = 1 - 1/K_i$. Under independent consideration, if the merged entity's consideration probability satisfies $1-\lambda_{jk} = (1-\lambda_j)(1-\lambda_k)$---i.e., the merged firm is considered whenever either predecessor would have been---then a merger between non-leader firms leaves the margin floor $\underline{\mu}$ unchanged but shifts posted distributions toward lower margins, raising $K_i$ for the remaining firms and, hence, the transaction-weighted pass-through. Even mergers with no immediate price effect can shift the incidence of future cost shocks toward consumers.

\section{Pass-Through Envelopes} \label{sec:envelopes}

In many empirical settings, we may not know the exact demand function. Perhaps we observe that demand is downward-sloping but cannot pin down its precise curvature. Or we may know that demand belongs to a particular family (e.g., linear, constant elasticity) but not the exact parameters. This section shows that even with such partial knowledge, we can still derive robust bounds on pass-through. These bounds provide ``worst-case'' scenarios for policy analysis, following the robust welfare approach of \citet{KangVasserman2025RobustWelfare}, and help identify when precise demand estimation is crucial versus when rough knowledge suffices.

\subsection{The Envelope Problem}

Instead of starting with a known demand function and computing pass-through, we ask: given a particular margin level \(\mu\) and cost \(c\), what is the range of possible pass-through rates across all admissible demand functions?

\begin{assumption}
Let \(\mathcal{X}\) denote the class of \emph{admissible demand functions} \(x\colon(0,1]\to\R_+\) that are continuous, weakly decreasing, and continuously differentiable, have strictly positive demand at the upper bound (\(x(1) > 0\)), and satisfy \Cref{ass:invertible}. We normalize \(x(1)=1\).
\end{assumption}

Formally, fix $\mu \in [0,1]$ and $c \in [0,1)$. The inverse problem is to find all prices $p$ consistent with
\[\tag{\(7\)} \label{eq:inverse}
(p-c)x(p) = \mu(1-c)x(1)\]
for some admissible demand function $x(\cdot)$.

\subsection{Universal Bounds}

We begin with the most general case: what can we say about pass-through knowing only that demand is downward-sloping? The answer provides universal bounds that apply regardless of the specific functional form.

\begin{theorem}\label{thm:universal}
For any admissible $x \in \mathcal{X}$, $\mu \in [0,1]$, and $c \in [0,1)$, we have the price bounds \(c \leq \phi(\mu,c) \leq c + \mu(1-c)\) and the pass-through bounds \(1-\mu \leq \phi_c(\mu,c)\). The upper bound for prices and the lower bound for pass-through are attained by unit demand.
\end{theorem}
These bounds have important economic implications. The pass-through bound $\tau \geq 1-\mu$ tells us that firms with lower market power (lower $\mu$) have pass-through rates bounded further from zero. In the limit, firms pricing at cost ($\mu = 0$) must have pass-through of at least 1, recovering the perfect competition result. Conversely, firms extracting maximum margins ($\mu$ near 1) could have pass-through rates approaching zero.
\begin{proof}
For price bounds, we normalize $x(1) = 1$ (by rescaling quantity units). Since $x$ is decreasing, we have $x(p) \geq 1$ for $p \in [c,1]$. From \eqref{eq:inverse},
\[
\mu(1-c) = (p-c)x(p) \geq p-c \quad \Longrightarrow \quad p \leq c + \mu(1-c).
\]
The lower bound $p \geq c$ is trivial.

For the pass-through bound, from \Cref{thm:quantile}, define $\varepsilon \equiv -(p-c)x'(p)/x(p) \geq 0$. From \Cref{ass:invertible}, $\frac{d}{dp}[(p-c)x(p)] = x(p)(1-\varepsilon) > 0$, so $\varepsilon < 1$. Then
\[
\phi_c(\mu,c) = \frac{1-\phi(\mu,c)}{(1-c)(1-\varepsilon)} \geq \frac{1-(c+\mu(1-c))}{1-c} = 1-\mu,
\] as $\varepsilon \geq 0$ and $\phi(\mu,c) \leq c + \mu(1-c)$.
\end{proof}

\subsection{Bounds for Specific Demand Families}

While universal bounds are useful, we can derive tighter bounds when we know more about demand. Different assumptions about demand curvature---whether demand is linear, exponential, or has constant elasticity---yield different pass-through bounds. These family-specific bounds help connect theoretical predictions to empirical demand estimation.

\begin{theorem}\label{thm:family_bounds}
Let \(\mu \in [0,1)\). For common demand families (with \(d\coloneqq1-c\)), in addition to the universal lower bound
\(1-\mu\le\phi_c(\mu,c)\), the following hold:
\begin{enumerate}
\item \textbf{Linear demand.} \(x(p)=1+b(1-p)\) with \(b\in[0,1/d)\):
\(\phi_c(\mu,c)\le\frac{1+\sqrt{1-\mu}}{2}\). The lower bound is attained at \(b=0\) (unit demand), and the upper envelope is tight as \(b\uparrow1/d\).

\item \textbf{Constant semi-elasticity.} \(x(p)=\mathrm e^{\beta(1-p)}\) with
\(\beta\in[0,1/d)\):
\(\phi_c(\mu,c)\le1\). The lower bound is attained at \(\beta=0\) (unit demand), and the upper envelope is tight as \(\beta\uparrow1/d\).

\item \textbf{Constant elasticity} \(x(p)=p^{-\eta}\) with \(0\le\eta<1/d\) and \(c>0\):
for \(\mu\in(0,1)\),
\[
\phi_c(\mu,c)=
\frac{1-\phi(\mu,c)}
{d\left(1-\eta\frac{\phi(\mu,c)-c}{\phi(\mu,c)}\right)},
\]
where \(\phi(\mu,c)\in(c,1)\) is the unique solution to
\[
\left(\phi(\mu,c)-c\right)\phi(\mu,c)^{-\eta}=\mu d.
\]
\end{enumerate}
\end{theorem}
For constant-elasticity demand, writing \(p\coloneqq\phi(\mu,c)\), we have
\[
\phi_c(\mu,c)>1\iff p<\eta d,\quad
\phi_c(\mu,c)=1\iff p=\eta d,\quad \text{and} \quad
\phi_c(\mu,c)<1\iff p>\eta d.
\]
In particular, when \(\eta=1\),
\[
\phi(\mu,c)=\frac{c}{1-\mu d},\qquad \text{and} \qquad
\phi_c(\mu,c)=\frac{1-\mu}{\left(1-\mu d\right)^2}.
\]

Figure~\ref{fig:passthrough_bounds} illustrates the linear envelope and a collection
of CES curves. In the linear case, pass-through lies between the universal lower
bound \(\tau=1-\mu\) (unit demand, \(b=0\)) and the upper envelope
\(\tau=\left(1+\sqrt{1-\mu}\right)/2\), which is tight as \(b\uparrow1/(1-c)\). Thus,
linear demand never generates over-shifting, and for \(\mu\in(0,1)\) pass-through is
strictly below one. For CES demand with \(c>0\), there is no analogous family-wide
upper envelope. Instead, pass-through is governed by the expression in
\Cref{thm:family_bounds}: under the maintained admissible range
\(\eta<1/(1-c)\), pass-through exceeds one exactly when the implied price satisfies
\(p<\eta(1-c)\). When the CES pass-through curve crosses the line \(\tau = 1\), therefore, depends jointly on \(\eta\), \(c\), and the induced price level, rather than on a universal ``critical elasticity.'' %Any plotted CES curve with \(\eta\ge1/(1-c)\), if shown, falls outside \Cref{ass:invertible} and should be interpreted only under the alternative normalization discussed immediately afterward.

\begin{figure}[t]
    \centering
    \includegraphics[width=\textwidth]{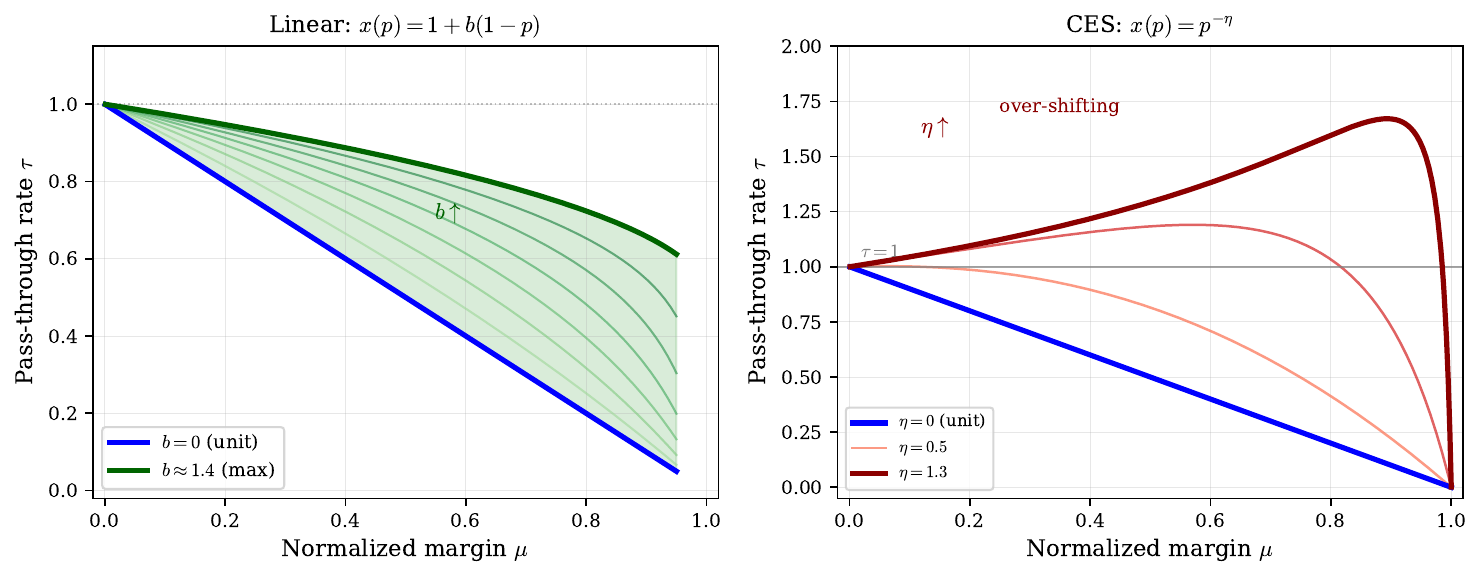}
    \caption{\textbf{Robust pass-through bounds by demand family.}}% \\ Left panel: For linear demand \(x(p)=1+b(1-p)\), pass-through lies between the unit-demand lower bound \(\tau=1-\mu\) and the upper envelope \(\tau=\left(1+\sqrt{1-\mu}\right)/2\), which is tight as \(b\uparrow1/(1-c)\). Linear demand, therefore, never generates over-shifting (\(\tau \leq 1\)). Right panel: For CES demand \(x(p)=p^{-\eta}\) with \(c>0\), pass-through is given by the formula in \Cref{thm:family_bounds}. Under the admissible range \(\eta<1/(1-c)\), pass-through exceeds one exactly when the implied price satisfies \(p<\eta(1-c)\). This crossing, therefore, depends jointly on \(\eta\), \(c\), and the induced price level---\emph{there is no universal critical elasticity}. Curves with \(\eta\ge1/(1-c)\) lie outside \Cref{ass:invertible} and are only illustrative.}
    \label{fig:passthrough_bounds}
\end{figure}

\section{Comparative Statics}\label{sec:cs}

The separation principle yields a simple logic for comparative statics. Any comparison of prices or pass-through---whether across market structures, across firms, or across demand specifications---reduces to comparing margin distributions and applying the maps $\phi$ and $\phi_c$. We develop this logic and its applications below.

The quantile pass-through formula (\Cref{thm:quantile}) expresses pass-through as a composition \(\tau^Q(u;c) = \phi_c(\mu(u), c)\). Any ordering of margins, therefore, translates into an ordering of prices and pass-through, subject to the monotonicity properties of $\phi$ and $\phi_c$.

\begin{proposition}\label{prop:ordering}
Let $\mu^A(u)$ and $\mu^B(u)$ be two margin quantile functions. If $\mu^B(u) \geq \mu^A(u)$ for all $u \in [0,1]$, then: \textbf{(a)} prices inherit the ordering, $p^B(u;c) \geq p^A(u;c)$ for all $u$; %, since $\phi(\mu,c)$ is increasing in $\mu$; 
and \textbf{(b)} pass-through inherits the ordering if $\phi_c$ is increasing in $\mu$, or the reverse ordering if $\phi_c$ is decreasing in $\mu$.
\end{proposition}

\begin{proof}
Part (a) follows from $\phi(\cdot,c)$ being strictly increasing (Lemma~\ref{lem:phi_exists}). Part (b) follows from the formula $\tau^Q(u;c) = \phi_c(\mu(u),c)$ and monotonicity of $\phi_c$.
\end{proof}

For unit demand, $\phi_c(\mu,c) = 1 - \mu$, which is decreasing in $\mu$. So higher margins imply lower pass-through. For linear demand, $\phi_c$ is also decreasing in $\mu$. For CES demand with sufficiently high elasticity, \(\phi_c\) need not be decreasing
in \(\mu\)---it can increase over parts of the support, so the pass-through ordering is ambiguous.

Consider first comparing the same firm across two consideration structures $A$ and $B$. If structure $B$ induces higher margins---say, because fewer consumers are informed or more are captive---then $B$ has higher prices at each quantile. If demand is not too convex (so $\phi_c$ is decreasing), then $B$ also has lower pass-through at each quantile: the less competitive market absorbs more of cost shocks.

The same logic applies to comparing different firms within a single market. Under independent consideration, \Cref{cor:margin_dominance} establishes that higher-reach firms maintain higher margins at each quantile. Likewise, in asymmetric duopoly, \Cref{prop:asymmetric} orders firms by their captive-to-reach ratios. From \Cref{prop:ordering}, these firms also charge higher prices at each quantile; and, if demand is not too convex, they also have lower pass-through at each quantile. Firms with more captive customers extract higher margins and absorb more of cost changes, partially insulating captive consumers.

Of course, for welfare analysis, transaction-weighted prices matter. The following corollary shows that price dominance carries through to mean paid prices.

\begin{corollary}\label{cor:tw_means_referee}
Fix \(c\in\left[0,1\right)\) and a firm \(i\).
Assume that in each market \(M\in\left\{A,B\right\}\) the posted-price quantile function \(p_i^M\left(\cdot;c\right)\) satisfies \(\underline m_i^M\left(c\right)\coloneqq \inf_{u\in\left[0,1\right]}\left(p_i^M\left(u;c\right)-c\right) > 0\).\footnote{This ensures that \(B_i^M\left(c\right)\coloneqq \int_0^1 \frac{1}{p_i^M\left(u;c\right)-c}du\) is finite.} If \(p_i^A\left(u;c\right)\ge p_i^B\left(u;c\right)\) for all \(u\in\left[0,1\right]\), then \(B_i^A\left(c\right)\le B_i^B\left(c\right)\).\end{corollary}

\begin{proof}Directly,
\[p_i^A\left(u;c\right)\ge p_i^B\left(u;c\right) \quad \Longrightarrow \quad \frac{1}{p_i^A\left(u;c\right)-c}\le \frac{1}{p_i^B\left(u;c\right)-c} \ \left(\forall \ u\right)  \quad \Longrightarrow \quad B_i^A\left(c\right)\le B_i^B\left(c\right),\]
as \(p\mapsto 1/\left(p-c\right)\) is strictly decreasing on \(\left(c,\infty\right)\).
\end{proof}

Together, these results provide a toolkit for comparative statics: the ordering principle (\Cref{prop:ordering}) translates margin comparisons into price and pass-through comparisons, whether across markets or across firms. Corollary~\ref{cor:tw_means_referee} extends price orderings to transaction-weighted objects relevant for welfare.

\section{Economic Applications} \label{sec:applications}

Incidence varies across the price distribution. The framework characterizes pass-through at each quantile, with patterns that depend on demand curvature. We illustrate with two markets featuring well-documented price dispersion.

\subsection{Gasoline Markets}

Retail gasoline exhibits substantial price dispersion even for a homogeneous product, with nearby stations often charging different prices. This dispersion reflects heterogeneous consumer search: some drivers compare prices across stations while others buy from the nearest option. Our framework applies directly.

The empirical pass-through literature has documented significant heterogeneity in gasoline markets. \citet{MarionMuehlegger2011FuelTaxIncidence} find that pass-through of state fuel taxes varies with supply conditions, while \citet{Stolper2017WhoBearsEnergyTaxes} shows that station-level pass-through varies with local competition and spatial isolation. \citet{MontagMamrakSagimuldinaSchnitzer2023} document heterogeneity at the consumer level, finding that informed consumers (who compare prices and buy cheap) face higher pass-through than uninformed consumers (who buy from one station). Our theory provides an explanation: within a single market, informed and uninformed consumers face a common posted-margin distribution but transact at different quantiles of it, and quantile pass-through varies across the distribution (\S\ref{sec:quantile}).

Under the framework, stations in competitive locations (where many consumers compare prices) operate at low margins. If demand is not too convex, these stations exhibit high pass-through. Stations in captive locations (highway exits, isolated areas) charge higher prices and, under moderate demand curvature, absorb more of cost shocks. The framework predicts that markets with more price-comparing consumers have higher aggregate pass-through when demand is not too convex, consistent with the empirical finding that pass-through is higher in urban areas with more stations.

\subsection{Online Retail}

Online markets feature persistent price dispersion despite low search costs. \citet{ellison2009search} show that firms actively obfuscate to soften price competition.\footnote{Related obfuscation mechanisms include add-on pricing \citep{ellison2005_model_addon_pricing} and consumer myopia \citep{gabaix_laibson2006_shrouded_attributes_consumer_myopia}.} Platforms shape consideration sets through search rankings, algorithmic recommendations, and sponsored listings---even when consumers can easily compare prices, they often consider only a subset of sellers.\footnote{Reference \citet{dinerstein_einav_levin_sundaresan2018_consumer_price_search_platform_design} on platform design and consumer search; \citet{nocke_rey2024_consumer_search_steering_choice_overload} on steering and choice overload; and \citet{teh_wang_watanabe2024_strategic_limitation_market_accessibility} on platforms strategically limiting accessibility. \citet{hagiu_jullien2011_why_intermediaries_divert_search} and \citet{teh_wright2022_intermediation_steering} formulate related models of intermediation and steering.}

\citet{Heim2021AsymmetricPassThrough} provides evidence linking search behavior to pass-through. Using data from price comparison sites, he finds that pass-through of input costs depends on consumer search intensity: cost increases pass through less when consumers search more. Our framework captures this mechanism: search intensity determines which part of the margin distribution consumers transact at, and hence which pass-through rate they face. For sufficiently convex demand, consumers transacting at lower margins face lower pass-through, consistent with \citeauthor{Heim2021AsymmetricPassThrough}'s finding.

The framework also applies to platform fees. Platforms charge sellers commissions. These fees enter as costs and pass through to consumers. When a platform shows each product to fraction $\lambda$ of users, the symmetric equilibrium has $\mu(u) = [(1-\lambda)/(1-\lambda u)]^{n-1}$. Platforms that increase $\lambda$ (showing more options) intensify competition, lowering margins and raising pass-through. Platforms that decrease $\lambda$ (curating selections) create captive segments, raising margins and lowering pass-through. The same platform design choices that affect price levels also affect who bears the platform's fees.

\section{Endogenous consideration and general-equilibrium pass-through}\label{subsec:endogenous-consideration}

Our baseline analysis treats the consideration structure \(\left\{\alpha_S\right\}_{S\subseteq N}\) as exogenous and invariant to the cost level \(c\). In many environments, however, consideration itself is an equilibrium outcome of marketing, placement, and product-design choices. A leading example is \citet{eliaz2011consideration,eliaz2011strategic}, in which firms choose marketing strategies jointly with payoff-relevant product characteristics, and these choices determine which products enter consumers' consideration sets. In such settings, a cost shock can shift the consideration structure and thereby shift the equilibrium distribution of normalized effective margins. This section records the resulting modification to our pass-through formulas in reduced form.

Let \(\Delta\left(2^N\right)\) denote the simplex over subsets of \(N\). We represent endogenous consideration by a reduced-form map
\[\begin{split}
    \alpha\colon[0,1)&\to\Delta\left(2^N\right)\\ c&\mapsto\left(\alpha_S\left(c\right)\right)_{S\subseteq N}.
\end{split}
\]
For each \(c\), the pricing game is played with marginal cost \(c\) and consideration structure \(\alpha\left(c\right)\). In the margin game, \Cref{thm:isomorphism} continues to imply that any equilibrium profile of \(\mu\)-CDFs \(\left(F_i^\mu\left(\cdot;c\right)\right)_{i\in N}\) depends on the demand system only through \(\alpha\left(c\right)\). Consequently, shifts in the consideration structure mediate any \(c\)-dependence of equilibrium \(\mu\)-quantile functions \(\mu_i\left(\cdot;c\right)\).

To separate the two channels, we define the \emph{conditional} (\textit{viz.}, competition-layer fixed) quantile pass-through rate as
\[
\tau_i^{Q,\mathrm{cond}}\left(u;c\right)\coloneqq\phi_c\left(\mu_i\left(u;c\right),c\right),
\]
which is exactly the object characterized in \Cref{thm:quantile}, evaluated at the equilibrium \(\mu\)-quantile \(\mu_i\left(u;c\right)\).

When consideration shifts with costs, the \emph{general-equilibrium} (total-derivative) quantile pass-through rate is
\[
\tau_i^{Q,\mathrm{GE}}\left(u;c\right)\coloneqq\frac{d}{dc}p_i\left(u;c\right).
\]
Assume \(c\mapsto\mu_i\left(u;c\right)\) is differentiable at the cost level of interest. Via the chain rule, we have the decomposition
\[
\tau_i^{Q,\mathrm{GE}}\left(u;c\right)=\tau_i^{Q,\mathrm{cond}}\left(u;c\right)+\phi_\mu\left(\mu_i\left(u;c\right),c\right)\mu_{i,c}\left(u;c\right),
\quad \text{where} \quad
\mu_{i,c}\left(u;c\right)\coloneqq\frac{d}{dc}\mu_i\left(u;c\right).
\]
In the baseline model with fixed \(\alpha\), we have \(\mu_{i,c}\left(u;c\right)=0\), so \(\tau_i^{Q,\mathrm{GE}}\left(u;c\right)=\tau_i^{Q,\mathrm{cond}}\left(u;c\right)\) and Definition \ref{def:7} coincides with the total derivative.

The new term depends on \(\phi_\mu\), which we can compute in closed form. Differentiating \eqref{eq:phi_def} with respect to \(\mu\) (holding \(c\) fixed), we obtain
\[
\phi_\mu\left(\mu,c\right)=\frac{\left(1-c\right)x\left(1\right)}{x\left(\phi\left(\mu,c\right)\right)+\left(\phi\left(\mu,c\right)-c\right)x'\left(\phi\left(\mu,c\right)\right)}.
\]
By \Cref{ass:invertible}, the denominator is strictly positive, so \(\phi_\mu\left(\mu,c\right)>0\). Evaluating at the equilibrium quantile, and using \(p_i\left(u;c\right)=\phi\left(\mu_i\left(u;c\right),c\right)\), we can write
\[
\phi_\mu\left(\mu_i\left(u;c\right),c\right)=\frac{\left(1-c\right)x\left(1\right)}{x\left(p_i\left(u;c\right)\right)+\left(p_i\left(u;c\right)-c\right)x'\left(p_i\left(u;c\right)\right)}.
\]
Similarly, \Cref{thm:quantile} implies
\[
\tau_i^{Q,\mathrm{cond}}\left(u;c\right)=\frac{x\left(p_i\left(u;c\right)\right)\left(1-p_i\left(u;c\right)\right)}{\left(1-c\right)\left[x\left(p_i\left(u;c\right)\right)+\left(p_i\left(u;c\right)-c\right)x'\left(p_i\left(u;c\right)\right)\right]}.
\]
Combining these expressions yields an explicit formula:
\[\begin{split}
    \tau_i^{Q,\mathrm{GE}}\left(u;c\right) =
&\frac{x\left(p_i\left(u;c\right)\right)\left(1-p_i\left(u;c\right)\right)}{\left(1-c\right)\left[x\left(p_i\left(u;c\right)\right)+\left(p_i\left(u;c\right)-c\right)x'\left(p_i\left(u;c\right)\right)\right]}\\
&+\frac{\left(1-c\right)x\left(1\right)}{x\left(p_i\left(u;c\right)\right)+\left(p_i\left(u;c\right)-c\right)x'\left(p_i\left(u;c\right)\right)}\mu_{i,c}\left(u;c\right).
\end{split}
\]
Thus, relative to the conditional formula in \Cref{thm:quantile}, endogenous consideration adds a term proportional to the cost-induced shift in the equilibrium \(\mu\)-quantile function. Because \(\phi_\mu\left(\mu,c\right)>0\), the sign of the additional term is governed by \(\mu_{i,c}\left(u;c\right)\), i.e., by whether a cost increase shifts the equilibrium distribution of normalized effective margins upward or downward.

In models where consideration is shaped by equilibrium marketing and product-design choices, such as \citet{eliaz2011consideration,eliaz2011strategic}, a cost change can shift \(\alpha\left(c\right)\) through changes in advertising intensity, platform placement, or product characteristics that affect consumers' attention and sampling. Our decomposition implies that any such mechanism affects pass-through only through the induced response \(\mu_{i,c}\left(u;c\right)\) of the competition layer, while the curvature layer \(\phi\) continues to govern the mapping from \(\left(\mu,c\right)\) to prices.

\begin{example}[Unit Demand]
    When \(x\left(p\right)\equiv 1\), we have \(\phi\left(\mu,c\right)=c+\left(1-c\right)\mu\), so
\[
p_i\left(u;c\right)=c+\left(1-c\right)\mu_i\left(u;c\right)
\]
and the GE quantile pass-through simplifies to
\[
\tau_i^{Q,\mathrm{GE}}\left(u;c\right)=1-\mu_i\left(u;c\right)+\left(1-c\right)\mu_{i,c}\left(u;c\right).
\]
In this benchmark, the full wedge between conditional and general-equilibrium pass-through is \(\left(1-c\right)\mu_{i,c}\left(u;c\right)\).
\end{example}

All of the demand-robust envelope results in \S\ref{sec:envelopes} apply directly to the conditional component \(\tau_i^{Q,\mathrm{cond}}\left(u;c\right)\), evaluated at the equilibrium price quantile \(p_i\left(u;c\right)\). The additional general-equilibrium term depends on the endogenous competition-layer response \(\mu_{i,c}\left(u;c\right)\), which demand curvature alone does not discipline.

 \appendix
\section{\S\ref{sec:isomorphism} Proofs}

\subsection{Proof of \texorpdfstring{\Cref{lem:phi_exists}}{Lemma \ref{ref{lem:phi_exists}}}}
\begin{proof}[Proof of \Cref{lem:phi_exists}]
Define $g(p) \equiv (p-c)x(p) - \mu(1-c)x(1)$ for $p \in [c,1]$. By the continuity of $x(\cdot)$, the function $g$ is continuous. Moreover,
\[g(c) = -\mu(1-c)x(1) \leq 0 \leq (1-c)x(1)(1-\mu) = g(1),\]
so by the intermediate value theorem, there exists $p^* \in [c,1]$ with $g(p^*) = 0$. Uniqueness follows because $g'(p) = x(p) + (p-c)x'(p) > 0$ on $[c,1]$ by \Cref{ass:invertible}, so $g$ is strictly increasing.
\end{proof}

\subsection{Proof of \Cref{lem:existence}}\label{app:existence-proof}

Let \(U_i\left(\mu\right) \coloneqq \mu_i q_i^{\mu}\left(\mu\right)\) denote firm \(i\)'s payoff.
\begin{proof}[Proof of \Cref{lem:existence}]
We verify the hypotheses of \citet[Theorem~5]{DasguptaMaskin1986DiscontinuousGames}.

\smallskip

\noindent\textbf{1. Compactness of the strategy sets and boundedness of the payoffs.} For each \(i\), the pure strategy set is \(A_i \coloneqq \left[0,1\right]\), a closed interval.
Moreover, for all \(\mu \in \left[0,1\right]^n\),
\[
0 \le q_i^{\mu}\left(\mu\right) \le \sum_{S \ni i} \alpha_S \le 1,
\quad \Longrightarrow \quad
0 \le U_i\left(\mu\right) = \mu_i q_i^{\mu}\left(\mu\right) \le 1,
\]
so \(U_i\) is bounded.

\noindent\textbf{2. Discontinuities occur only on diagonals.}
Fix \(i\). If \(\mu\) satisfies \(\mu_i \neq \mu_j\) for all \(j \neq i\), then in each set \(S \ni i\)
either \(i \notin M_S\left(\mu\right)\) or else \(M_S\left(\mu\right)=\left\{i\right\}\), and this classification
is locally constant in a neighborhood of \(\mu\). Hence, \(q_i^{\mu}\left(\cdot\right)\) is locally constant
and \(U_i\left(\cdot\right)\) is continuous at such \(\mu\). Therefore, \(U_i\) can be discontinuous only
when \(\mu_i = \mu_j\) for some \(j \neq i\), i.e., only on a union of diagonal hyperplanes.

Equivalently, in the notation of \citet[Equation~2]{DasguptaMaskin1986DiscontinuousGames}, we may take
\(D(i)=n-1\) and, for each \(j \neq i\), define the one-to-one continuous function
\(f_{ij}^{1}\left(a_i\right) \coloneqq a_i\). Then
\[
A^{*}\left(i\right) = \left\{ \mu \in \left[0,1\right]^n \colon \exists j \neq i \text{ such that } \mu_j = f_{ij}^{1}\left(\mu_i\right) = \mu_i \right\},
\]
and the set of discontinuities of \(U_i\) is a subset of \(A^{*}\left(i\right)\).

\noindent\textbf{3. The sum of payoffs is continuous (therefore, upper semi-continuous).}
For any \(S \neq \emptyset\), uniform tie-breaking implies that the total payoff generated by mass
\(\alpha_S\) equals \(\alpha_S m_S\left(\mu\right)\):
\[
\sum_{i \in S} \mu_i  \alpha_S \frac{\mathbf{1}\left\{ i \in M_S\left(\mu\right)\right\}}{\left|M_S\left(\mu\right)\right|}
=
\alpha_S \, m_S\left(\mu\right) \sum_{i \in M_S\left(\mu\right)} \frac{1}{\left|M_S\left(\mu\right)\right|}
=
\alpha_S m_S\left(\mu\right).
\]
Summing over all \(S \subseteq N\) yields
\[
\sum_{i \in N} U_i\left(\mu\right)
=
\sum_{S \subseteq N, S \neq \emptyset} \alpha_S m_S\left(\mu\right).
\]
Because each \(m_S\left(\mu\right)=\min_{j\in S}\mu_j\) is continuous and there are finitely many
sets \(S\), the sum \(\sum_{i \in N} U_i\left(\mu\right)\) is continuous.

\noindent\textbf{4. \(U_i\) is weakly lower semi-continuous in \(\mu_i\) \citep[Definition~6]{DasguptaMaskin1986DiscontinuousGames}.}
Fix \(i\) and fix \(\mu_{-i}\). For each \(S \ni i\), define \(m_{S,-i} \coloneqq \min_{j \in S\setminus\left\{i\right\}} \mu_j\), with the convention \(m_{\left\{i\right\},-i} \coloneqq 1\). Consider the one-variable function
\(t \mapsto U_i\left(t,\mu_{-i}\right)\) on \(\left[0,1\right]\).

We claim that for every \(t_0 \in \left(0,1\right]\), \(\liminf_{t \nearrow t_0} U_i\left(t,\mu_{-i}\right) \ge U_i\left(t_0,\mu_{-i}\right)\). To see this, fix \(S \ni i\) and examine the \(S\)-contribution to \(U_i\).
If \(t_0 < m_{S,-i}\), then for all \(t\) sufficiently close to \(t_0\) from below we still have \(t < m_{S,-i}\), so \(i\) is the unique minimizer in \(S\) and the \(S\)-contribution equals
\(\alpha_S t\), which is continuous at \(t_0\).
If \(t_0 > m_{S,-i}\), then for all \(t\) sufficiently close to \(t_0\) from below we still have \(t > m_{S,-i}\), so \(i \notin M_S\left(\cdot\right)\) and the \(S\)-contribution is identically \(0\)
near \(t_0\). If \(t_0 = m_{S,-i}\), then at \(t_0\) firm \(i\) is tied for the minimum in \(S\), so the
\(S\)-contribution at \(t_0\) equals
\[
\alpha_S t_0 \frac{1}{\left|M_S\left(t_0,\mu_{-i}\right)\right|}
\le \alpha_S t_0.
\]
For any \(t < t_0\), however, \(i\) becomes the unique minimizer in \(S\), so the \(S\)-contribution equals \(\alpha_S t\). Consequently,
\[
\liminf_{t \nearrow t_0} \alpha_S t = \alpha_S t_0 \ge \alpha_S t_0 \frac{1}{\left|M_S\left(t_0,\mu_{-i}\right)\right|}.
\]
In all cases, the \(S\)-contribution satisfies the desired lower-semicontinuity inequality, \textit{viz.}, \(\liminf_{t \nearrow t_0} U_i(t,\mu_{-i}) \geq U_i(t_0,\mu_{-i})\) for all \(t_0 \in (0,1]\);
and summing over \(S \ni i\) yields
\(\liminf_{t \nearrow t_0} U_i\left(t,\mu_{-i}\right) \ge U_i\left(t_0,\mu_{-i}\right)\). This left-limit inequality establishes \citet[Definition~6]{DasguptaMaskin1986DiscontinuousGames} with \(\lambda=1\) (full weight on the left-hand \(\liminf\)) for all \(t_0 \in (0,1]\).

At the left endpoint, \(t_0=0\), we have \(U_i\left(0,\mu_{-i}\right)=0\) and \(U_i \ge 0\),
so \(\liminf_{t \searrow 0} U_i\left(t,\mu_{-i}\right) \ge 0 = U_i\left(0,\mu_{-i}\right)\), so the condition holds trivially. Thus, \(U_i\) satisfies \citet[Definition~6]{DasguptaMaskin1986DiscontinuousGames} for all \(t_0 \in [0,1]\).%,and with the endpoint convention at \(t_0=0\).

\noindent\textbf{5. Apply \citet[Theorem~5]{DasguptaMaskin1986DiscontinuousGames}.}
Steps 1-4 verify all hypotheses of \citet[Theorem~5]{DasguptaMaskin1986DiscontinuousGames}. Therefore, the
\(\mu\)-game possesses a mixed-strategy equilibrium.
\end{proof}

\subsection{Proof of \Cref{thm:isomorphism}}\label{app:isomorphism-proof}

\begin{proof}[Proof of \Cref{thm:isomorphism}]
Recall that firm $i$'s profit at price $p$ is
\[
\Pi_i(p;c) = (p-c)x(p)q_i(p) = (1-c)x(1) \cdot \mu(p;c) \cdot q_i(p),
\]
with the factor $(1-c)x(1)$ being constant across all firms and prices. Recall also firm \(i\)'s normalized effective margin at price \(p\):
\[\mu(p;c) \equiv \frac{(p-c)x(p)}{(1-c)x(1)} \in [0,1],\]
and the inverse map \(\phi(\mu,c)\), which solves
\[\label{eq:appinverse}\tag{\(A1\)}(\phi(\mu,c) - c)x(\phi(\mu,c)) = \mu(1-c)x(1).\]

Setting, in turn, \(\mu = 0\) and \(\mu = 1\) in \eqref{eq:appinverse} yield $\phi(0,c) = c$ and $\phi(1,c) = 1$. Moreover \(\phi\) is increasing in \(\mu\). Hence, there is a bijection between price-CDFs $F_i$ on $[c,1]$ and $\mu$-CDFs $F_i^\mu$ on $[0,1]$ via \(F_i^\mu(\mu) = F_i(\phi(\mu,c))\).

For a profile \(\mu \in \left[0,1\right]^n\) and a nonempty set \(S \subseteq N\), define
\[
m_S\left(\mu\right) \coloneqq \min_{j \in S} \mu_j,
\quad \text{and} \quad
M_S\left(\mu\right) \coloneqq \left\{ j \in S \colon \mu_j = m_S\left(\mu\right) \right\}.
\]
Under our uniform tie-breaking stipulation, consumers with consideration set \(S\) allocate their mass \(\alpha_S\) equally across the minimizers \(M_S\left(\mu\right)\).

Recall that firm \(i\)'s demand when posting price \(p\) is
\[q_i(p) =\sum_{S\ni i}\alpha_S\cdot
\mathbb{E}\left[
\frac{\mathbf{1}\left\{i\in M_S\left(p, p_{-i}\right)\right\}}
{\left|M_S\left(p, p_{-i}\right)\right|}
\right].
\]
When firm $i$ posts $\mu$ and rivals use $(F_j^\mu)_{j \neq i}$, the demand share becomes
\[\tag{\(A2\)}\label{eq:appdemand}
q_i^{\mu}\left(\mu\right)
\coloneqq
\sum_{S \ni i} \alpha_S \cdot \mathbb E \left[\frac{\mathbf{1}\left\{ i \in M_S\left(\mu, \mu_{-i}\right) \right\}}{\left|M_S\left(\mu, \mu_{-i}\right)\right|}\right].
\]

The equilibrium indifference condition in the original game requires constant profit on the support:
\[
(p-c)x(p)q_i(p) = \pi_i(c), \quad \forall p \in \supp(F_i).
\]
Dividing by $(1-c)x(1)$:
\[
\mu(p;c) \cdot q_i(p) = \frac{\pi_i(c)}{(1-c)x(1)} \equiv \tilde{\pi}_i, \quad \forall p \in \supp(F_i).
\]
Since $\phi$ is strictly increasing, the lowest-price firm in each consideration set is also the lowest-margin firm, so $q_i(p) = q_i^\mu(\mu(p;c))$. Under the bijection, the indifference condition becomes:
\[
\mu \cdot q_i^\mu(\mu) = \tilde{\pi}_i, \quad \forall \mu \in \supp(F_i^\mu),
\]
which is precisely the equilibrium condition for a game with:
\begin{itemize}
\item Unit demand (quantity fixed at \(1\));
\item zero marginal cost;
\item ``price'' $\mu \in [0,1]$; and
\item demand shares computed via \eqref{eq:appdemand} using $\mu$-distributions.
\end{itemize}
Moreover, the transformed equilibrium conditions depend only on $\{\alpha_S\}$ through the demand-share formula. Neither $c$ nor $x(\cdot)$ appears in the $\mu$-space equilibrium, though they determine the mapping $\phi$ back to prices.

To complete the bijection, we verify that any $\mu$-equilibrium induces an equilibrium in the original game. Suppose $(F_i^\mu)_{i \in N}$ is an equilibrium in the $\mu$-game with constant profits $\tilde{\pi}_i$ on the support. Then, for any $p$ in the support of $F_i$:
\[\begin{split}
    \Pi_i(p;c) &= (p-c)x(p)q_i(p) \\
&= (1-c)x(1) \cdot \mu(p;c) \cdot q_i(p) \\
&= (1-c)x(1) \cdot \mu(p;c) \cdot q_i^\mu(\mu(p;c)) \\
&= (1-c)x(1) \cdot \tilde{\pi}_i
\end{split}\]
where the third equality uses that $q_i(p) = q_i^\mu(\mu(p;c))$ by construction, and the fourth uses that $\mu(p;c)$ is in the support of $F_i^\mu$ when $p$ is in the support of $F_i$.

For any $p \in [c,1]$ outside the support of $F_i$, the corresponding $\mu(p;c)$ is outside the support of $F_i^\mu$. Since the map $p \mapsto \mu(p;c)$ is strictly increasing (by \Cref{ass:invertible}), each deviation in price space corresponds to a unique deviation in $\mu$-space. Therefore, the ``no profitable deviation'' inequality is preserved: if $\mu \cdot q_i^\mu(\mu) \leq \tilde{\pi}_i$ for $\mu$ off support, then $(p-c)x(p)q_i(p) \leq (1-c)x(1)\tilde{\pi}_i = \pi_i(c)$ for the corresponding $p$ off support. Thus, the induced price distributions form an equilibrium with profits $\pi_i(c) = (1-c)x(1)\tilde{\pi}_i$.

Existence of an equilibrium in the pricing game follows from \Cref{lem:existence} combined with the bijection established here.
\end{proof}

\section{\S\ref{sec:equilibrium} and \S\ref{sec:quantile} Proofs}

\subsection{Proof of \Cref{prop:symmetric}}

\begin{proof}[Proof of \Cref{prop:symmetric}]
By the $\mu$-isomorphism (\Cref{thm:isomorphism}), the margin game is strategically equivalent to the unit-demand pricing game of \citet{ArmstrongVickers2022Patterns}. Their Proposition~1 establishes that with symmetric consideration structure, the unique symmetric equilibrium has interval support $\left[\underline{\mu},1\right]$ with $\underline{\mu} = \rho$, equilibrium profit $\pi^* = \rho\sigma$, and no atoms on $\left(\underline{\mu},1\right)$.

It remains to derive the quantile function. Let $G(\mu) \equiv 1 - F^\mu(\mu)$. Since the distribution is atomless on the interior, the demand share at $\mu$ is
\[
q^\mu(\mu) = \sum_{S \ni i} \alpha_S G(\mu)^{|S|-1}.
\]
Indifference on $[\rho,1]$ requires $\mu \cdot q^\mu(\mu) = \rho\sigma$. Dividing by $\sigma$ and defining
\[
H(s) \equiv \frac{1}{\sigma}\sum_{S\ni i}\alpha_S s^{|S|-1},
\]
we obtain $\mu \cdot H(G(\mu)) = \rho$, so $\mu = \rho / H(G(\mu))$. With the quantile transformation $u = F^\mu(\mu) = 1 - G(\mu)$, this yields
\(\mu(u) = \frac{\rho}{H(1-u)}\). The support endpoints are confirmed by $H(1) = 1$ and $H(0) = \rho$, delivering $\mu(0) = \rho$ and $\mu(1) = 1$.
\end{proof}

\subsection{Proof of \Cref{prop:asymmetric}}
\begin{proof}[Proof of \Cref{prop:asymmetric}]
Any duopoly has symmetric interactions, so uniqueness follows from \citet[Proposition~1]{ArmstrongVickers2022Patterns}. We derive the CDFs from the indifference conditions.

First we pin down equilibrium profits. Firm 1 guarantees profit $\alpha_1$ by pricing at $\mu = 1$ (serving only captives). Firm 2 benefits from Firm 1's high price floor ($\underline{\mu} = \rho_1$), guaranteeing rents above its captive share:
\[
\pi_1^* = \alpha_1, \quad \text{and} \quad \pi_2^* = \underline{\mu}(\alpha_2 + \alpha_{12}) = \rho_1(\alpha_2 + \alpha_{12}).
\]
Note that $\pi_2^* > \alpha_2$ since $\rho_1 > \rho_2$. Moreover, from $\rho_i = \alpha_i/(\alpha_i + \alpha_{12})$, we have $\frac{\alpha_i}{\alpha_{12}} = \frac{\rho_i}{1-\rho_i}$ and $\frac{\alpha_i + \alpha_{12}}{\alpha_{12}} = \frac{1}{1-\rho_i}$.

Next, we use the two firms' indifference conditions to back out the cdfs. When firm 1 posts $\mu$ and firm 2 uses CDF $F_2^\mu$, firm 1's profit is:
\[
\mu[\alpha_1 + \alpha_{12}(1 - F_2^\mu(\mu))] = \alpha_1 \quad \Longrightarrow \quad
1 - F_2^\mu(\mu) = \frac{\rho_1}{1-\rho_1}\left(\frac{1-\mu}{\mu}\right).
\]

When firm 2 posts $\mu$ and firm 1 uses CDF $F_1^\mu$, firm 2's profit is:
\[
\mu[\alpha_2 + \alpha_{12}(1 - F_1^\mu(\mu))] = \pi_2^* = \rho_1(\alpha_2 + \alpha_{12})
 \quad \Longrightarrow \quad
1 - F_1^\mu(\mu) = \frac{1}{1-\rho_2} \left( \frac{\rho_1}{\mu} - \rho_2 \right).
\]%\frac{\alpha_2+\alpha_{12}}{\alpha_{12}} \left( \frac{\rho_1}{\mu} - \rho_2 \right)

Finally, we pin down the support. At the lower bound, $F_1^\mu(\underline{\mu}) = 0$ requires $\frac{\rho_1}{\underline{\mu}} - \rho_2 = 1-\rho_2$, yielding $\underline{\mu} = \rho_1$. At the upper bound, $F_2^\mu(1^{-}) = 1$ confirms firm 2 mixes continuously to 1. However, $F_1^\mu(1^{-}) = 1 - \frac{\rho_1 - \rho_2}{1-\rho_2} < 1$, confirming the atom $\Delta_1$.\end{proof}

\subsection{Proof of \Cref{prop:n_asymmetric}}

\begin{proof}[Proof of \Cref{prop:n_asymmetric}]
Independent consideration satisfies symmetric interactions \citep[Section~3]{ArmstrongVickers2022Patterns}, so their Proposition~1 guarantees a unique equilibrium with common lower bound $\underline{\mu}=\rho_1$, nested interval supports, profits $\pi_j^\ast = \lambda_j\underline{\mu}$, and the common-CDF property $\lambda_j F_j^\mu(\mu) = \lambda_k F_k^\mu(\mu) \eqqcolon \Gamma(\mu)$ for all firms active at $\mu$. We verify the claimed $\mu$-space CDFs.

When firms $\{1,\dots,m\}$ are active at $\mu$ (i.e., $\mu \in [\bar{\mu}_{m+1}, \bar{\mu}_m]$), each active firm $j \le m$ has $\lambda_j F_j^\mu(\mu) = \Gamma(\mu)$ and each inactive firm $i > m$ has $F_i^\mu(\mu) = 1$. Firm $j$'s profit is
\[
\Pi_j(\mu) = \mu\,\lambda_j\,C_m\,(1-\Gamma(\mu))^{m-1},
\]
where $C_m = \prod_{h=m+1}^n(1-\lambda_h)$. Setting $\Pi_j(\mu) = \lambda_j\underline{\mu}$ and solving yields
\[
(1-\Gamma(\mu))^{m-1} = \frac{\underline{\mu}}{\mu\,C_m}, \qquad \text{so} \qquad \Gamma(\mu) = 1 - \left(\frac{\underline{\mu}}{\mu\,C_m}\right)^{1/(m-1)}.
\]

The cutoff $\bar{\mu}_k$ where firm $k$'s CDF reaches 1 satisfies $\Gamma(\bar{\mu}_k) = \lambda_k$, yielding
\[
\bar{\mu}_k = \frac{\prod_{h=2}^{k-1}(1-\lambda_h)}{(1-\lambda_k)^{k-2}}.
\]
By taking $\mu \uparrow 1$ on the top interval (where $m=2$) we obtain $\Gamma(1^-) = \lambda_2$, so $F_1^\mu(1^-) = \lambda_2/\lambda_1 < 1$, confirming firm 1's atom $\Delta_1 = 1 - \lambda_2/\lambda_1$.

For the no-deviation check, if $\mu < \underline{\mu}$ then $\Pi_j(\mu) = \mu\lambda_j < \underline{\mu}\lambda_j = \pi_j^\ast$. If $j > m$ at some $\mu \in [\bar{\mu}_{m+1}, \bar{\mu}_m]$, then $\Pi_j(\mu) = \lambda_j\underline{\mu}(1-\Gamma(\mu))/(1-\lambda_j) \le \pi_j^\ast$ since $\Gamma(\mu) \ge \lambda_{m+1} \ge \lambda_j$.\end{proof}

\subsection{Proof of \Cref{thm:quantile}}
\begin{proof}[Proof of \Cref{thm:quantile}]
Recall the implicit definition \eqref{eq:phi_def}:
\[
(\phi(\mu,c) - c)x(\phi(\mu,c)) = \mu(1-c)x(1).
\]
Differentiating both sides with respect to $c$,\footnote{Permitted by the implicit function theorem, by the smoothness from \Cref{ass:demand,ass:invertible}.} holding $\mu$ fixed, and rearranging yields 
\[\phi_c[x(\phi) + (\phi-c)x'(\phi)] = x(\phi) - \mu x(1).\]
From \eqref{eq:phi_def}, we have $\mu x(1) = \frac{(\phi-c)x(\phi)}{1-c}$. Substituting this in produces
\[\phi_c[x(\phi) + (\phi-c)x'(\phi)] = x(\phi)\frac{1-\phi}{1-c},\]
so,
\[
\phi_c(\mu,c) = \frac{x(\phi(\mu,c))(1-\phi(\mu,c))}{(1-c)[x(\phi(\mu,c)) + (\phi(\mu,c)-c)x'(\phi(\mu,c))]}.
\]
Since $\mu_i(u)$ does not depend on $c$ (\Cref{thm:isomorphism}\ref{item:invariance}), setting $\mu = \mu_i(u)$ and noting that $\phi(\mu_i(u),c) = p_i(u;c)$ yields \eqref{eq:passthrough}.
\end{proof}

\section{\S\ref{sec:transaction} Transaction-Weighted Proofs and Extensions}
\label{app:transaction_proofs}

\subsection{Supporting Lemmas}

\begin{lemma}\label{lem:transaction} At equilibrium, for all \(p \in \supp(F_i)\) with \(p > c\), \(T_i(p;c) = x(p)q_i(p) = \frac{\pi_i(c)}{p-c}\).
\end{lemma}

\begin{proof}
Immediate from the equilibrium indifference condition.
\end{proof}

\begin{lemma} \label[lemma]{lem:mean_paid}
The mean transaction-weighted price is \(\bar{p}_i^{\trans}(c) = c + \frac{1}{B_i(c)}\), where $B_i(c) \equiv \int (p-c)^{-1} dF_i(p;c)$.
\end{lemma}
\begin{proof}
The \emph{transaction-weighted CDF} is:
\[F_i^{\trans}(p;c) = \frac{\int_{c}^p \frac{1}{s-c} dF_i(s;c)}{\int_{c}^1 \frac{1}{s-c} dF_i(s;c)},\] so the mean paid price is
\[\bar{p}_i^{\trans}(c) = \int p dF_i^{\trans}(p;c) = \frac{\int p \cdot \frac{1}{p-c} dF_i(p;c)}{\int \frac{1}{p-c} dF_i(p;c)} = \frac{1 + cB_i(c)}{B_i(c)} = c + \frac{1}{B_i(c)}.\qedhere\]
\end{proof}

\subsection{Proof of \Cref{prop:transaction}}
\begin{proof}[Proof of \Cref{prop:transaction}]
From \Cref{lem:mean_paid},
\[
\tau_i^{\trans}(c)
=
\frac{d\bar{p}_i^{\trans}(c)}{dc}
=
1-\frac{B_i'(c)}{B_i(c)^2}.
\]

We need to compute \(B_i'(c)\). By the \(\mu\)-isomorphism, \(\mu_i(\cdot)\) is
fixed as \(c\) varies, so
\[
B_i(c)
=
\int_0^1 \frac{1}{p_i(u;c)-c}\,du
=
\int_0^1 \frac{1}{\phi(\mu_i(u),c)-c}\,du.
\]
Differentiating under the integral sign is justified by the support assumption:
locally around \(c\), \(p_i(u;\tilde c)-\tilde c\) is bounded below uniformly in
\(u\), and \(\phi_c(\mu_i(u),\tilde c)\) is bounded. Consequently,
\[
B_i'(c)
=
\int_0^1
\frac{\partial}{\partial c}
\left(
\frac{1}{p_i(u;c)-c}
\right)du  =
-\int_0^1
\frac{\tau_i^Q(u;c)-1}{(p_i(u;c)-c)^2}\,du  =
-\int_0^1
\frac{\phi_c(\mu_i(u),c)-1}
{\left(\phi(\mu_i(u),c)-c\right)^2}
\,du,
\]
using \(p_i(u;c)=\phi(\mu_i(u),c)\) and
\(\tau_i^Q(u;c)=\phi_c(\mu_i(u),c)\). Substituting into the expression for
\(\tau_i^{\trans}(c)\) yields \eqref{eq:trans_passthrough}.
\end{proof}
\iffalse
\begin{proof}[Proof of \Cref{prop:transaction}]
From \Cref{lem:mean_paid},
\[
\tau_i^{\trans}(c) = \frac{d\bar{p}_i^{\trans}(c)}{dc} = 1 - \frac{B_i'(c)}{B_i(c)^2}
\]

We need to compute $B_i'(c)$. We take the quantile representation
\[
B_i(c) = \int_0^1 \frac{1}{p_i(u;c) - c} du,
\]
and differentiate under the integral sign (valid by dominated convergence when $\rho > 0$, which ensures $p_i(u;c) - c \geq \epsilon > 0$ uniformly):
\[B_i'(c) = \int_0^1 \frac{\partial}{\partial c}\left(\frac{1}{p_i(u;c) - c}\right) du = \int_0^1 \frac{-(\tau_i^Q(u;c) - 1)}{(p_i(u;c) - c)^2} du = -\int_0^1 \frac{\phi_c(\mu_i(u),c) - 1}{(\phi(\mu_i(u),c) - c)^2} du,\]
using $p_i(u;c) = \phi(\mu_i(u),c)$ and $\tau_i^Q(u;c) = \phi_c(\mu_i(u),c)$. Substituting into the expression for $\tau_i^{\trans}(c)$ yields \eqref{eq:trans_passthrough}.
\end{proof}\fi

\subsection{Proof of \Cref{cor:trans_asymmetric}}

\begin{proof}
We have
\[B_i(c)=\int_0^1 \frac{1}{(1-c)\mu_i(u)} du=\frac{K_i}{1-c}.\]
Thus, using \Cref{lem:mean_paid}, $\bar p_i^{\trans}(c)=c+(1-c)/K_i$. Differentiating, $\tau_i^{\trans}=1-1/K_i$.
\end{proof}

\subsection{Proofs of \Cref{prop:aggregate_margin} and \Cref{cor:industry_passthrough}}\label{app:aggregate_margin} % and develops comparative statics.

\iffalse
Let \(\mathcal{T}\coloneqq\left\{\text{at least one firm is considered}\right\}\), so \(\mathcal{T}\) is equivalently the event that a transaction occurs. Let
\[
\bar{\mu}^{\trans}\coloneqq\mathbb{E}\left[\mu^{\trans}\mid\mathcal{T}\right]
\]
denote the aggregate expected transaction margin: the share-weighted average of the per-firm transaction-weighted margins,
\[
\bar{\mu}^{\trans}=\sum_{i=1}^n s_i\bar{\mu}_i^{\trans},
\]
where \(s_i\) is firm \(i\)'s share of all transactions in the market. Let
\[
\Lambda\coloneqq\sum_{i=1}^n\lambda_i.
\]
Recall from \Cref{prop:n_asymmetric} that under independent consideration the common lower support bound is
\[
\underline{\mu}=\prod_{h=2}^n\left(1-\lambda_h\right).
\]
\fi
Denote
\[
\bar{\mu}^{\trans}
\coloneqq
\mathbb{E}\!\left[
\frac{(p-c)x(p)}{(1-c)x(1)}
\;\middle|\;
\text{a purchase occurs}
\right],
\]
where \(p\) is the price paid by a randomly selected purchasing consumer. Thus,
\(\bar{\mu}^{\trans}\) is the average normalized effective margin among purchasing consumers. Also denote
\(\Lambda\coloneqq\sum_{i=1}^n\lambda_i\) and recall from \Cref{prop:n_asymmetric} that under independent consideration the
common lower support bound is
\(\underline{\mu}=\prod_{h=2}^n\left(1-\lambda_h\right)\).
\begin{proof}[Proof of \Cref{prop:aggregate_margin}]
If some non-leaders have equal reaches, the payoff formula used below follows by
continuity from the (generic) strict-order case. Thus it suffices to write the argument for the strict-order case covered by \Cref{prop:n_asymmetric}.

Let \(\pi_j^\mu\) denote firm \(j\)'s equilibrium payoff in the margin game. From \Cref{prop:n_asymmetric}\eqref{it:452},
\(\pi_j^\mu=\lambda_j\underline{\mu}\). Moreover, equilibrium payoffs in the pricing game are \((1-c)x(1)\) times equilibrium payoffs in the margin game (\Cref{thm:isomorphism}). Consequently, firm \(j\)'s equilibrium profit in the pricing game is
\(\pi_j(c)
=
(1-c)x(1)\lambda_j\underline{\mu}\); and by summing across firms, total expected industry profit is \(\sum_{j=1}^n \pi_j(c)
=
(1-c)x(1)\underline{\mu}\Lambda\).

The same total expected industry profit can be computed from purchases. A purchase occurs with probability \(1-\alpha_\emptyset\), and conditional on a
purchase, the average normalized effective margin is \(\bar{\mu}^{\trans}\).
Therefore,
\[
\sum_{j=1}^n \pi_j(c)
=
(1-c)x(1)
\left(1-\alpha_\emptyset\right)
\bar{\mu}^{\trans}.
\]
Equating the two expressions, we get
\[
\left(1-\alpha_\emptyset\right)\bar{\mu}^{\trans}
=
\underline{\mu}\Lambda.
\]
Under independent consideration,
\[
\alpha_\emptyset
=
\prod_{j=1}^n\left(1-\lambda_j\right)
=
\left(1-\lambda_1\right)\underline{\mu}.
\]
Substituting gives
\[
\bar{\mu}^{\trans}
=
\frac{\underline{\mu}\Lambda}{1-\alpha_\emptyset}
=
\frac{\underline{\mu}\Lambda}
{1-\left(1-\lambda_1\right)\underline{\mu}},
\]
which is exactly \eqref{eq:aggregate_margin}.\end{proof}

The aggregate margin formula immediately yields \Cref{cor:industry_passthrough}.

\begin{proof}[Proof of \Cref{cor:industry_passthrough}]
Under unit demand, the normalized effective margin of a purchase at price \(p\) is \(\frac{p-c}{1-c}\). Thus, since \(\bar p^{\trans}(c)\) is the average price paid conditional on purchase,
we use the definition of \(\bar{\mu}^{\trans}\) to obtain
\[\bar{\mu}^{\trans}
=
\frac{\bar{p}^{\trans}(c)-c}{1-c} \quad \iff \quad 
\bar{p}^{\trans}(c)=c+\left(1-c\right)\bar{\mu}^{\trans}.
\]
\eqref{eq:aggregate_margin} tells us that \(\bar{\mu}^{\trans}\) depends only on the reach vector and not on \(c\), so
\[
\tau^{\trans,\mathrm{agg}}
=
\frac{d\bar{p}^{\trans}(c)}{dc}
=
1-\bar{\mu}^{\trans}.
\]
Substituting \eqref{eq:aggregate_margin} yields \eqref{eq:industry_passthrough}.\end{proof}

\section{\S\ref{sec:envelopes} Pass-Through Envelope Proofs}

\subsection{Proof of \Cref{thm:family_bounds}}\label{app:family-bounds-proof}

\begin{proof}
Let \(d \coloneqq 1-c\). Fix \(\mu\in(0,1)\), write \(p\coloneqq\phi(\mu,c)\), \(m\coloneqq p-c\), and \(q\coloneqq1-p\), so \(m,q\in(0,d)\) and \(m+q=d\). From \(\phi(0,c)=c\) and \(\phi(1,c)=1\), we have \(\phi_c(0,c)=1\) and \(\phi_c(1,c)=0\). The universal lower bound \(1-\mu\le\phi_c(\mu,c)\) is \Cref{thm:universal}, so we focus on the family-specific upper envelopes and formulas.

\medskip
\noindent\textbf{Case 1: Linear demand.}
Let \(x(p)=1+b(1-p)\) with \(b\in[0,1/d)\). When \(b=0\), we recover unit demand, so \(\phi_c(\mu,c)=1-\mu\). For \(b>0\), the effective-margin equation
\[\label{eqd1}\tag{\(D1\)}m\left(1+b(d-m)\right)=\mu d,\]
is solved to get
\[m=\frac{(1+bd)-\sqrt{(1+bd)^2-4b\mu d}}{2b},\]
where we take the smaller root to keep \(p\le1\).

Since \(x'(p)=-b\), \Cref{thm:quantile} delivers
\[
\phi_c(\mu,c)=\frac{x(p)(1-p)}{d\left[x(p)+(p-c)x'(p)\right]}
=\frac{(1+bq)q}{d(1+bq-bm)}.
\]
Using the equilibrium relation \(b=\left(\mu d-m\right)/(mq)\), this simplifies to
\[
\phi_c(\mu,c)=\frac{\mu q^2}{(m-\mu d)^2+\mu(1-\mu)d^2}.
\]
Differentiating with respect to \(m\) yields
\[
\frac{\partial\phi_c}{\partial m}
=-\frac{2\mu(1-\mu)dm(d-m)}
{\left[(m-\mu d)^2+\mu(1-\mu)d^2\right]^2}<0.
\]
Thus, \(\phi_c\) is strictly decreasing in \(m\). Implicit differentiation of \eqref{eqd1} produces
\[
\frac{\partial m}{\partial b}
=-\frac{m(d-m)}{1+bd-2bm}
=-\frac{mq}{1+bd-2bm}<0,
\]
where the denominator is positive by \Cref{ass:invertible}. Therefore,
\(\phi_c\) is increasing in \(b\). The supremum over admissible \(b<1/d\) is the limit as \(b\uparrow1/d\). In that weak-boundary case, \eqref{eqd1} becomes
\(m\left(2-\frac{m}{d}\right)=\mu d\),
so
\(m=d\left(1-\sqrt{1-\mu}\right)\).

Substituting into the pass-through formula, \eqref{eq:passthrough} in \Cref{thm:quantile}, yields
\[
\sup_{b<1/d}\phi_c(\mu,c)=\frac{1+\sqrt{1-\mu}}{2}.
\]

\smallskip
\noindent\textbf{Case 2: Constant semi-elasticity.}
Let \(x(p)=e^{\beta(1-p)}\) with \(\beta\in[0,1/d)\). When \(\beta=0\), we again
recover unit demand, so \(\phi_c(\mu,c)=1-\mu\). For \(\beta>0\), the inverse
equation is
\[
(p-c)e^{\beta(1-p)}=\mu d.
\]
Because \(x'(p)=-\beta x(p)\), \Cref{thm:quantile} produces
\[
\phi_c(\mu,c)=\frac{x(p)(1-p)}{d\left[x(p)+(p-c)x'(p)\right]}
=\frac{1-p}{d\left(1-\beta(p-c)\right)}
=\frac{q}{d(1-\beta m)}.
\]
Hence,
\[
\phi_c(\mu,c)\le1
\iff q\le d(1-\beta m)
\iff d-m\le d-d\beta m
\iff m(1-d\beta)\ge0,
\]
which holds for every \(\beta<1/d\). Thus, \(\phi_c(\mu,c)\le1\). The bound is
tight because in the weak-boundary case \(\beta=1/d\),
\[
\phi_c(\mu,c)=\frac{q}{d(1-m/d)}=1
\]
for every \(\mu\in(0,1)\), so \(\sup_{\beta<1/d}\phi_c(\mu,c)=1\).

\smallskip
\noindent\textbf{Case 3: Constant elasticity.}
Let \(x(p)=p^{-\eta}\) with \(0\le\eta<1/d\) and \(c>0\). The effective-margin
equation is
\[
(p-c)p^{-\eta}=\mu d.
\]
Moreover,
\[
\frac{d}{dp}\left[(p-c)p^{-\eta}\right]
=p^{-\eta-1}\left[p-\eta(p-c)\right],
\]
so admissibility is equivalent to \(p-\eta(p-c)>0\) on the relevant support.

Applying \Cref{thm:quantile} produces
\[
\phi_c(\mu,c)=
\frac{p^{-\eta}(1-p)}
{d\left[p^{-\eta}-\eta(p-c)p^{-\eta-1}\right]}
=\frac{1-p}{d\left(1-\eta\frac{p-c}{p}\right)}.
\]
Subtracting one and simplifying,
\[
\phi_c(\mu,c)-1
=\frac{(p-c)(\eta d-p)}
{d\left[p-\eta(p-c)\right]}.
\]
Since the denominator is positive on the interior support, for \(p\in(c,1)\) we
have
\[
\phi_c(\mu,c)>1\iff p<\eta d,\qquad
\phi_c(\mu,c)=1\iff p=\eta d,\qquad
\phi_c(\mu,c)<1\iff p>\eta d.
\]

In the special case \(\eta=1\),
\[
(p-c)p^{-1}=\mu d\Longrightarrow p=\frac{c}{1-\mu d}.
\]
Substituting into the pass-through formula, \eqref{eq:passthrough} in \Cref{thm:quantile}, yields
\[
\phi_c(\mu,c)=\frac{1-\mu}{\left(1-\mu d\right)^2}.
\]
At \(\eta=0\), the CES family collapses to unit demand and recovers the lower bound \(\phi_c(\mu,c)=1-\mu\).\end{proof}

\bibliography{references}
\bibliographystyle{plainnat}

\newpage

\section{Supplementary Appendix}

\subsection{\S\ref{sec:transaction} Further Discussion}

Three features deserve emphasis.

\begin{remark}
\label{rem:attention_decomposition}
\Cref{prop:aggregate_margin} is just an accounting identity for expected industry profit:
\[
\sum_{j=1}^n\pi_j(c)=\left(1-c\right)x(1)\underline{\mu}\Lambda
=\left(1-c\right)x(1)\left(1-\alpha_\emptyset\right)\bar{\mu}^{\trans}.
\]
The left-hand side computes expected industry profit by summing firm-level equilibrium payoffs. The right-hand side computes the same object by averaging realized transaction margins. In particular, the proof of \Cref{prop:aggregate_margin} uses only the margin-game payoff identity. The unit-demand assumption enters in \Cref{cor:industry_passthrough}, when normalized margins are translated back into transaction prices.
\end{remark}

\begin{remark}
\label{rem:sufficient_statistic}
The conditional aggregate margin is determined by three market-level numbers: the support floor \(\underline{\mu}\), the total reach \(\Lambda\), and the no-purchase mass \(\alpha_\emptyset\) (equivalently, \(\lambda_1\)). This is still a substantial dimension reduction relative to the full reach vector \(\left(\lambda_1,\ldots,\lambda_n\right)\), but unlike the industry-profit identity above, the conditioned object cannot be recovered from \(\left(\underline{\mu},\Lambda\right)\) alone.
\end{remark}

\begin{remark}
\label{rem:pic_comparative}
\Cref{prop:aggregate_margin} yields sharp signs for the conditional aggregate margin. The leader's reach raises the average margin paid by transacting consumers:
\[
\frac{\partial\bar{\mu}^{\trans}}{\partial\lambda_1}
=\frac{\underline{\mu}\left[1-\underline{\mu}\left(1+\Lambda-\lambda_1\right)\right]}{\left[1-\left(1-\lambda_1\right)\underline{\mu}\right]^2}>0,
\]
where the sign follows from
\[
\frac{1}{\underline{\mu}}=\prod_{h=2}^n\left(1+\frac{\lambda_h}{1-\lambda_h}\right)\ge 1+\sum_{h=2}^n\frac{\lambda_h}{1-\lambda_h}>1+\sum_{h=2}^n\lambda_h=1+\Lambda-\lambda_1.
\]
Each non-leader's reach lowers it:
\[
\frac{\partial\bar{\mu}^{\trans}}{\partial\lambda_k}
=\frac{\underline{\mu}}{\left[1-\left(1-\lambda_1\right)\underline{\mu}\right]^2}\left(1-\left(1-\lambda_1\right)\underline{\mu}-\frac{\Lambda}{1-\lambda_k}\right)<0
\]
for \(k\neq 1\), and the sign follows because
\[
1-\left(1-\lambda_1\right)\underline{\mu}=1-\alpha_\emptyset\le \Lambda<\frac{\Lambda}{1-\lambda_k}.
\]
Conditioning on transaction therefore yields a clean asymmetry: greater leader reach pushes the average paid margin up, whereas greater non-leader reach pushes it down.
\end{remark}

\end{document}